\documentclass[aip,reprint,floatfix,amsmath,amssymb]{revtex4-2}

\usepackage{dcolumn}% Align table columns on decimal point
\usepackage{bm}% bold math
\usepackage[utf8]{inputenc}
\usepackage[T1]{fontenc}
\usepackage{diffcoeff}
\usepackage{amsthm}
\usepackage{graphicx}
\usepackage{float}
\usepackage{overpic}
\usepackage{tikz}
\usepackage[hidelinks]{hyperref}
\usepackage[nameinlink,capitalize]{cleveref}
\usepackage{csquotes}
\usepackage{xspace}
\usepackage{booktabs}

\newcommand*{\naturals}{\ensuremath{\mathbb{N}}}
\newcommand*{\integers}{\ensuremath{\mathbb{Z}}}
\newcommand*{\reals}{\ensuremath{\mathbb{R}}}
\newcommand*{\torus}{\ensuremath{\mathbb{T}}}
\newcommand*{\rationals}{\ensuremath{\mathbb{Q}}}
\newcommand*{\dt}[1]{\ensuremath{\dot{#1}}}
\newcommand*{\abs}[1]{\ensuremath{\lvert #1 \rvert}}
\newcommand*{\sign}{\ensuremath{\operatorname{sgn}}}
\newcommand*{\Z}{\integers}

\newcommand*{\R}{\reals}

\newcommand*{\half}{\frac{1}{2}}
\newcommand*{\coupling}{\ensuremath{P_s}}
\newcommand*{\adaptation}{\ensuremath{\mathcal{A}}}
\newcommand*{\transienttime}{\ensuremath{T_\text{transient}}}
\newcommand*{\stabletime}{\ensuremath{T_\text{stable}}}
\newcommand*{\fixedpointtolerance}{\ensuremath{\delta_\text{f.p.}}}
\newcommand*{\quiescent}{Q\xspace}
\newcommand*{\quiescentlibration}{Q\textsubscript{\( \ell \)}\xspace}
\newcommand*{\spiking}{S\xspace}
\xspaceaddexceptions{\quiescent \quiescentlibration \spiking}
\newcommand*{\SNone}{SN\textsubscript{1}\xspace}
\newcommand*{\SNonea}{SN\textsubscript{1a}\xspace}
\newcommand*{\SNoneb}{SN\textsubscript{1b}\xspace}
\newcommand*{\SNtwo}{SN\textsubscript{2}\xspace}
\newcommand*{\LPCone}{LPC\textsubscript{1}\xspace}
\newcommand*{\LPCtwo}{LPC\textsubscript{2}\xspace}
\usepackage{color}

\newtheorem{lemma}{Lemma}[section]

\newcounter{panel}

\newcommand*{\panellabel}[1]{\refstepcounter{panel}\alph{panel}\label{#1}}
\crefname{panel}{panel}{panels}

\begin{document}

\preprint{AIP/123-QED}

\title[Co-evolutionary dynamics for two adaptively coupled Theta neurons]{Co-evolutionary dynamics for two adaptively coupled Theta neurons}% Force line breaks with \\
% \thanks{Footnote to title of article.}

\author{Felix Augustsson}

\author{Erik A. Martens}
% \homepage{http://erikmartens.net}
\email{erik.martens@math.lth.se}
%  \affiliation[Also at ]{Physics Department, XYZ University.}%Lines break automatically or can be forced with \\

\date{\today}% It is always \today, today,
\affiliation{%
${}^a$Centre for Mathematical Sciences, Lund University, S\"olvegatan 18, 221 00 Lund, Sweden\\ 
}%

% An article usually includes an abstract, a concise summary of the work covered at length in the main body of the article. It is used for secondary publications and for information retrieval purposes. 

\begin{abstract}
Natural and technological networks exhibit dynamics that can lead to complex cooperative behaviors, such as synchronization in coupled oscillators and rhythmic activity in neuronal networks. Understanding these collective dynamics is crucial for deciphering a range of phenomena from brain activity to power grid stability. 
Recent interest in co-evolutionary networks has highlighted the intricate interplay between dynamics on and of the network with mixed time scales. 
Here, we explore the collective behavior of excitable oscillators in a simple networks of two Theta neurons with adaptive coupling without self-interaction.
Through a combination of bifurcation analysis and numerical simulations, we seek to understand how the level of adaptivity in the coupling strength, $a$, influences the dynamics.
We first investigate the dynamics possible in the non-adaptive limit; our bifurcation analysis reveals stability regions of quiescence and spiking behaviors, where the spiking frequencies mode-lock in a variety of configurations. Second, as we increase the adaptivity $a$, we observe a widening of the associated Arnol'd tongues, which may overlap and give room for multi-stable configurations. For larger adaptivity, the mode-locked regions may further undergo a period-doubling cascade into chaos. 
Our findings contribute to the mathematical theory of adaptive networks and offer insights into the potential mechanisms underlying neuronal communication and synchronization.
\end{abstract}

%Use showkeys class option if keyword
\maketitle

\begin{quotation}
Networks appear in a variety of natural and technological systems connecting dynamic units, such as heart and nerve cells in the body, and generators and consumers in the power grid. The network coupling facilitates emergent collective behaviors, such as the synchronization of oscillatory units. Examples include neurons adhering to collective firing during epileptic seizure, and the frequency locking of alternating currents in the power grid. Maintaining (or losing) such collective behaviors is critical for the proper functioning of these systems.
Much research effort has already been laid in studying systems with either dynamics \emph{on} the network nodes, or \emph{of} the network edges. Combining dynamics \emph{on} and \emph{of} the networks results in co-evolutionary or adaptive networks, leading to new exciting challenges in understanding and analyzing the resulting dynamics. 
To make progress, this relatively young research area in mathematics has focused on relatively simple phase oscillator models. 
Here, we aim to extend the understanding of such networks by considering \emph{excitable} oscillators, namely Theta neurons with adaptive coupling. We find several dynamic phenomena intrinsic to adaptive coupling, including not only the widening of mode locking regions, where neurons spike with different frequency ratios; but also their overlapping in multi-stable regions, and a period-doubling cascade into chaos. Thus, the system offers a wide range of stable mode-locking configurations for neural spiking behavior.
\end{quotation}

% % % % % 

\section{Introduction}

% general topic, complex networks, or neural networks
Networks are ubiquitous in nature and technology, including a range of systems from small to large scales~\cite{strogatz2001exploring}, such as metabolic cell circuits~\cite{BarabasiOltvai2004}, vascular and leaf venation networks~\cite{blinder2013cortical,kirst2020mapping,roth2001evolution,ronellenfitsch2015topological}, neural networks in the brain~\cite{sporns2013structure,lynn2019physics}, power grid networks~\cite{pagani2014power,rohden2012self}, the internet~\cite{tu2000robust}, and transport systems~\cite{Kaluza2010,xie2007measuring}. While for some purposes, networks may be considered stationary, the inclusion of dynamic processes are of great significance in a variety of contexts. Dynamics may occur on the network nodes (dynamics \emph{on} the network) or along the edges forming the network (dynamics \emph{of} the network). Examples of dynamics \emph{on} the nodes include neurons, heart cells, flashing fireflies, Josephson junctions, metronomes and pendulum clocks~\cite{strogatz2004sync,PikvskyRosenblum2003synchronization}; dynamics \emph{of} the network includes processes in social networks, synaptic plasticity in the brain, regulation of flow in vascular networks, and more~\cite{Kempter1999,Jacobsen2008,MartensKlemm2017,mestre2020cerebrospinal}.
Dynamics on networks are known to lead to emergent cooperative or collective behaviors, such as swarming and flocking in fish and birds~\cite{ginelli2010relevance,couzin2009collective}, or the synchronization of frequencies or phases in networks of coupled oscillators~\cite{strogatz2004sync,rosenblum2003synchronization,PikvskyRosenblum2003synchronization} and collective firing of neurons~\cite{luke2013complete,montbrio2015macroscopic}, which have been subject to a long series of studies~\cite{strogatz2000kuramoto,acebron2005kuramoto,bick2020understanding,martens2024integrability}. 

While these types of dynamics have been studied for a few decades, the combination of both in so-called \emph{co-evolutionary} or \emph{adaptive networks} has attracted increased interest more recently~\cite{gross2008adaptive,berner2023adaptive}. The complex interactions that may occur between two high dimensional systems which could be coupled in a variety of ways and at mixed time scales~\cite{kuehn2015multiple} represents a challenge in terms of developing a mathematical theory. It thus makes sense to start studying this topic using simple models. Some effort has already been undertaken in attempting to understand the dynamics of simple phase oscillator models, such as Kuramoto model~\cite{kuramoto1984chemical} where oscillators are rotating constantly and are coupled with first order harmonics, combined with coupling strengths that adapt in response to the phase difference between two individual nodes~\cite{seliger2002plasticity,berner2023adaptive,Berner2019,berner2021desynchronization,juettner2023adaptive,CestnikMartensAdaptive2024}.

In this study, we ask the question how \emph{excitable oscillators}, such as neurons, behave collectively in an adaptive network. While such dynamics has been investigated computationally in a variety of settings~\cite{YamakouJinjieMartens2024}, this topic is subject for developing further rigid mathematical analysis.
To take first steps in the direction of answering this question, we choose a simple phase oscillator model, the Theta neuron model~\cite{ermentrout1986parabolic,gutkin2022theta,ermentrout2008ermentrout}, combined with a simple adaptation rule for the coupling strength. The Theta neuron is not only a good candidate model because of its simplicity, but is also equivalent to the quadratic integrate-and-fire model, can be derived from higher dimensional neuronal models~\cite{ermentrout1986parabolic,pietras2019network}, featuring exact dimensional reduction methods exist for populations of Theta neurons with uniform coupling, akin to neural masses,~\cite{luke2013complete,bick2020understanding,coombes2018next,martens2024integrability,pietras2023exact}, and has been used to study a variety of applications in a neuroscientific contexts~\cite{bick2020understanding,schmidt2018network,coombes2018next}.

To simplify the analysis, we follow a similar paradigm as in previous work by one of the authors~\cite{juettner2023adaptive}, and study the dynamics for two adaptively coupled Theta neurons; we parameterize the strength of adaptation with the parameter $a$, starting from the non-adaptive limit of coupled neurons, and ask the questions: what dynamics occur for increasing adaptivity, which can genuinely be attributed to adaptive dynamics? Does adaptivity strengthen or weaken the emergent behaviors such as collective spiking? 
To obtain some answers we set up the model in Sec.~II and begin the study in Sec.~III  with a careful bifurcation analysis for the dynamics of non-adaptively coupled neurons, for non-identical (Sec.~III.C) and identical excitability parameters (Sec.~III.D). Understanding these limiting cases, we are ready to begin a numerical analysis for the system with adaptive coupling. We conclude the study with a discussion of the obtained results in Sec.~IV.

% {\it 
%   Overview of what is known vs unknown $\to$ lead in to the topic pursued in this ms. Your question here (in general terms)!
% 
% }

% ========================================================================
\section{Model}
\subsection{Network of Theta neurons}
We consider a model of $N$ interacting Theta neurons, where the phase \( \theta_k(t) \in \torus = S^1 \simeq \R/2\pi \Z \) of the $k$th neuron, 
\(k\in[N]:=\{1,\ldots,N\}\), evolves according to 
\begin{equation}\label{eq:theta_neuron}
  \dt{\theta}_k = 1 - \cos{\theta_k} + \left( 1 + \cos{\theta_k} \right) \iota_k,
\end{equation}
with parameter \(\iota_k = \eta_k+I_k\) where \(\eta_k\) is the excitability threshold and \(I_k\) the synaptic input current.
The Theta neuron~\eqref{eq:theta_neuron} is the normal form of the saddle-node-on-invariant-circle (SNIC) bifurcation~\cite{ermentrout2008ermentrout} and is a canonical type 1 neuron~\cite{ermentrout1986parabolic}.
For \(\iota_k<0\), a stable and unstable fixed point (\(\theta_-,\theta_+\)) occur on the phase circle \(\torus\); for \(\iota_k = 0\), these fixed points coalesce in a saddle-node bifurcation; for \(\iota_k>0\), the flow on the circle results in a periodic motion (spiking).

If \(\iota_k<0\), the Theta neuron is said to be \emph{excitable} or \emph{quiescent}: in the absence of perturbations, the phase relaxes to the stable fixed point \( \theta_- \) on the phase circle \(\torus\); however, a perturbation may lead to a single spike (at \(\theta=\pi\)) before returning to the stable fixed point.
This may happen in at least two ways: (i) the phase is perturbed across the unstable fixed point (constituting a threshold), this is possible, considering that the Theta model derives from a higher dimensional model~\cite{ermentrout1986parabolic} so that the circle is embedded in a higher dimensional space; (ii) a very short-lived (time scale of a single cycle) increase in \( \iota_k \) momentarily pushes the system across the bifurcation threshold  \( \iota_k = 0 \).
If \(\iota_k > 0\), the neuron is excited and is said to be \emph{spiking} or \emph{firing} (periodically).

% Model, part 2:
In a network of Theta neurons, the input current may result from a variety of interactions~\cite{bick2020understanding}.
Here, we assume that the input current for neuron \( k \) is given by a sum of pulses emitted from adjacent neurons \(l\),
\begin{equation}
  I_k = \frac{1}{d(k)} \sum_{l \in [N]} \kappa_{kl} \coupling(\theta_l),
\end{equation}
with pulses~\cite{ariaratnam2001phase,bick2020understanding} given by
\begin{align}
  \coupling(\theta) = a_s \left( 1 - \cos{\theta} \right)^s,
\end{align}
where the shape parameter \(s\in N\) controls the pulse width, and the normalization constant \( a_s=2^s(s!)^2/(2s)! \) satisfies \( \int_0^{2\pi}P_s(\theta) \, \dl \theta = 2\pi \).
The pulse function \( \coupling \) will attain its maxima at \( \theta = \pi \), which is referred to as a \emph{spike}.
Pulses received from neuron \( l \) are weighted by an synaptic interaction or coupling strength \( \kappa_{kl} \), normalized by the number  \(d(k)\) of neurons incident to neuron \(k\).

In this study, we assume that neurons be \emph{autaptic}, i.e., neurons have no self-interactions with \(\kappa_{ll}=0 \) for \( l\in[N]\), and  \(d(k)=n-1\).

\subsection{Co-evolutionary network model}
Many models assume that the coupling strengths, \( \kappa_{kl} \), be constant, in which case we say the coupling is \emph{non-adaptive}.
In this study, however, our aim is to study the effects of adaptive coupling strengths, i.e., coupling strengths that adapt according to the activity of neurons.
We consider the case where the coupling adapts on a time scale \( \varepsilon^{-1} \)  according to the rule
\begin{equation}\label{eq:adaptationrule1}
  \varepsilon^{-1}\dt{\kappa}_{kl} = \adaptation(\theta_k, \theta_l) - \kappa_{kl},
\end{equation}
with an \emph{adaptation function} \( \adaptation = \adaptation(\theta_k,\theta_l)\).
Thus, the adaptation is implemented to be pairwise or local between two neurons \( k \) and \( l \), and is homogeneous across the network.
Combining Eqs.~\eqref{eq:theta_neuron} and \eqref{eq:adaptationrule1} results in adaptive or \emph{co-evolutionary} network dynamics~\cite{gross2008adaptive,berner2023adaptive}.
The second term in \eqref{eq:adaptationrule1} guarantees that the coupling strength stays bounded as long as \( A \) is bounded.
In this study, we choose the adaptive function
\begin{equation}
  \adaptation(\theta_k, \theta_l) = b + a \cos(\theta_k - \theta_l + \beta),
\end{equation}
where \( b \) implements a \emph{baseline} for the coupling strength, \( a \) is the \emph{adaptivity} or \emph{adaptation strength}, and \( \beta \) an \emph{adaptation shift}.

% Explain basic behavior of adaptation function:
% When the phase difference \( \theta_l - \theta_k \) matches the offset \( \beta \), the coupling strength \( \kappa_{kl} \) relaxes over time to \( b \).
Non-adaptive dynamics is retrieved when \( a = 0 \).
In this case the coupling converges to the baseline value \( \kappa_{kl} \to b \) as \(\ t\to \infty\).
% The relaxation also happens in the special case of the adaptation strength \( a = 0 \), which therefore asymptotically is equivalent to using a non-adaptive model where all coupling strengths are homogeneously parameterized as \( \kappa_{kl} = b \).
Varying the adaptation shift tunes the type of adaptation.
For \(\beta=0,\pi\), the adaptation function is symmetric, \(\adaptation(\theta_k,\theta_l)=\adaptation(\theta_l,\theta_k)\).
For \( \beta = 0 \) and \(a > 0\), the coupling for in-phase configurations, i.e., neurons with phase difference close to 0 (or \( \pi \)), is amplified (or suppressed).
For \(\beta=\pi\) and \(a > 0\) (which is equivalent to \( \beta = 0 \) and \(a < 0\)), the coupling for anti-phase configurations, i.e., neurons with phase difference close to \(\pi\) (or 0), is amplified (or suppressed).
Other choices for $\beta$ may break the symmetry of the system, resulting in more complicated dynamics~\cite{juettner2023adaptive}.

\subsection{Governing equations for \texorpdfstring{\( N=2 \)}{N = 2} neurons with adaptation}
For the sake of this study, to simplify the analysis we restrict our focus to \( N=2 \) neurons with pulse shape parameter \(s = 1\).
Further, we choose \( \beta = 0 \) (which also covers \( \beta = \pi \) if the sign of \( a \) is flipped), which implies that the difference coupling \( \Delta := \kappa_{21} - \kappa_{12} \to 0\) as $t\to\infty$, since we then have \(  \dt{\Delta} = - \Delta \).
Thus \( \kappa := (\kappa_{12} + \kappa_{21}) / 2 \) is the only dynamic inter-neuron coupling strength in the system (if transients are ignored).
Using the assumption of autaptic neurons, we then have \( I_k = \kappa(1 - \cos(\theta_k)) \).
The dynamics then obeys the \emph{governing equations}
\begin{subequations} \label{eq:reduced-model}
  \begin{align}
    \dt{\theta_1} &= 1 - \cos\theta_1 + \left( 1 + \cos\theta_1 \right) \left( \eta_1 + \kappa \left( 1 - \cos\theta_2 \right) \right), \label{eq:reduced-model-1} \\
    \dt{\theta_2} &= 1 - \cos\theta_2 + \left( 1 + \cos\theta_2 \right) \left( \eta_2 + \kappa \left( 1 - \cos\theta_1 \right) \right), \label{eq:reduced-model-2} \\
    \dt{\kappa} &= \varepsilon \left(b + a \cdot \cos(\theta_2 - \theta_1) - \kappa \right). \label{eq:reduced-model-3}
  \end{align}
\end{subequations}
We note that one expects that trajectories in the asymptotic time limit of  \eqref{eq:reduced-model} reside in the compact space given by
\begin{equation}
  \left( \theta_1, \theta_2, \kappa \right) \in \torus^2 \times \left[ b - \abs{a}, b + \abs{a} \right],
\end{equation}
since the range of \( cos \) lies in the interval \( [-1, 1] \).

In what follows, we assume that the adaptation evolves with \( \varepsilon = 0.01 \) and has a zero base line (\( b = 0 \)).

% ========================================================================

\section{Analysis for non-adaptive coupling \texorpdfstring{\( (a = 0) \)}{(a = 0)}}
\subsection{Fixed point analysis and saddle-node bifurcations} \label{sec:nonadaptive-fp-analysis}
The fixed point conditions for \eqref{eq:reduced-model} are
\begin{subequations} \label{eq:non-adaptive-fp}
  \begin{align}
    0 &= 1 - \cos\theta_1 + \left( 1 + \cos\theta_1 \right) \left( \eta_1 + \kappa \left( 1 - \cos\theta_2 \right) \right), \label{eq:non-adaptive-fp-1} \\
    0 &= 1 - \cos\theta_2 + \left( 1 + \cos\theta_2 \right) \left( \eta_2 + \kappa \left( 1 - \cos\theta_1 \right) \right). \label{eq:non-adaptive-fp-2}
  \end{align}
\end{subequations}
By the change of variables, \( \cos \theta_1 = 2x - 1 \), \( \cos \theta_2 = 2y - 1 \) and \( \nu_k = \eta_k + 2 \kappa - 1 \) (where \( 0\leq x\leq 1\) and \( 0\leq y\leq 1 \)), the fixed point equations in \cref{eq:non-adaptive-fp} can be expressed as
\begin{subequations} \label{eq:non-adaptive-fp-simplified}
  \begin{align}
  \label{eq:non-adaptive-fp-simplified-a}
    0 &= -2 \kappa x y + \nu_1 x + 1, \\ 
  \label{eq:non-adaptive-fp-simplified-b}
    0 &= -2 \kappa x y + \nu_2 y + 1.
  \end{align}
\end{subequations}
However, from \cref{eq:non-adaptive-fp-simplified} we see that solutions with \( x = 0 \) or \( y = 0 \) are not valid solutions, so for any fixed point \( x,y \in (0, 1] \).

We wish to eliminate one variable.
Subtracting \eqref{eq:non-adaptive-fp-simplified-a} from \eqref{eq:non-adaptive-fp-simplified-b} implies that
\begin{equation} \label{eq:non-adaptive-fp-simplified-condition}
  \nu_1 x = \nu_2 y.
\end{equation}
Thus, solutions can only exist if \( \sign \nu_1 = \sign \nu_2 \).
This condition can be satisfied in two ways:

(i) If \( \nu_1 = \nu_2 = 0 \), Eqs.~\eqref{eq:non-adaptive-fp-simplified} reduce to the condition 
\begin{equation}
  2 \kappa xy = 1.
\end{equation}
This condition defines a one-dimensional manifold of fixed points as long as \( \kappa > \half \), and represents the (zero-dimensional) point \( x = y = 1 \) when \( \kappa = \half \).
On the other hand, if $\kappa<\half$ there is no fixed point.
Note that $\nu_1=\nu_2 =0 $ implies \( \eta_1 = \eta_2 = 1-2\kappa\),  i.e., neurons are identical.
This curve is seen in Fig.~\ref{fig:nonadaptive-sym-attractors} (dashed line).

(ii) If \( \nu_1,\nu_2 \neq 0 \), the condition  \( \sign \nu_1 = \sign \nu_2 \) means that either
\( \nu_1, \nu_2 > 0 \) or \( \nu_1, \nu_2 < 0 \).
We may without loss of generality assume that \( \abs{\nu_1} \leq \abs{\nu_2} \).
From \eqref{eq:non-adaptive-fp-simplified-condition} it follows that
\begin{equation}
  y = \frac{\nu_1}{\nu_2} x,
\end{equation}
implying that \( y \in (0, 1] \) if \( x \in (0, 1] \).
Considering Eqs.~\eqref{eq:non-adaptive-fp-simplified}, it is therefore sufficient in the case of non-zero, same sign \( \nu_1, \nu_2 \) to solve the quadratic equation,
\begin{equation}\label{eq:fixedpointquadratic}
  0 = -2 \kappa \frac{\nu_1}{\nu_2} x^2 + \nu_1 x + 1,
\end{equation}
for \( x \in (0, 1] \).
% Derive fold bifurcation curves based on (13)
Saddle-node bifurcations occur in two scenarios.
The first is when \( x = \half(1 + \cos{\theta_1}) = 1 \) (or \( y = 1 \) in the case of \( \abs{\nu_2} \leq \abs{\nu_1} \)).
In the original parameters, this condition is equivalent to
\begin{subequations}\label{eq:SNone}
  \begin{align}
    0 &= \eta_1 \eta_2 + 2 \eta_2 \kappa - \eta_1, \text{ for } \eta_2 \leq 0 \\
    0 &= \eta_1 \eta_2 + 2 \eta_1 \kappa - \eta_2, \text{ for } \eta_1 \leq 0
  \end{align}
\end{subequations}
The second is when the discriminant of the quadratic polynomial in Eq.~\eqref{eq:fixedpointquadratic} is 0 while the solution to the equation \( x \in (0, 1] \) (or \( y \in (0, 1] \) in the case of \( \abs{\nu_2} \leq \abs{\nu_1} \)), i.e.
\begin{equation}\label{eq:SNtwo}
  \begin{split}
    (\eta_1 + 2 \kappa - 1) (\eta_2 + 2 \kappa - 1) + 8 \kappa = 0&, \\
    \text{ for } \quad 1 - 2 \abs{\kappa} \leq \eta_1 \leq 1 + 2 \abs{\kappa}& \\
    \text{ and } \quad 1 - 2 \abs{\kappa} \leq  \eta_2 \leq 1 + 2 \abs{\kappa}&.
  \end{split}
\end{equation}
The resulting saddle-node bifurcation curves \SNone (Eq.~\eqref{eq:SNone}) and \SNtwo (Eq.~\eqref{eq:SNtwo}) are shown as black solid curves in Figs.~\ref{fig:nonadaptive-sym-attractors} and \ref{fig:nonadaptive-asym-attractors}.

% Explain table 1 with number of fixed points
It is furthermore possible to determine the number of fixed points that appear in the different regions outlined by the fold bifurcations.
A derivation is given in Appendix~\ref{app:polynomial-in-interval}, and the results are summarized in Table~\ref{tab:nonadaptive-sym-fixed-points}~and~\ref{tab:nonadaptive-asym-fixed-points}.

In the following two sections, we specialize these results to discuss the cases of identical and non-identical excitabilities separately, see Fig.~\ref{fig:nonadaptive-sym-attractors} and \ref{fig:nonadaptive-asym-attractors}.

\subsection{Dynamics for identical neurons\texorpdfstring{ (\( \eta_1 = \eta_2 \))}{}}

We consider first the situation when the two neurons are identical, i.e., \( \eta:=\eta_1 = \eta_2 \), and uncoupled, so that \( \kappa = b = 0 \) in the governing equations.
We then have \( \iota_k = \eta \) for \( k = 1,2 \), and each neuron effectively behaves like a single neuron.
Recall that we already explained in Sec.~\ref{sec:nonadaptive-fp-analysis} following~\eqref{eq:non-adaptive-fp}, such a neuron exhibits a SNIC bifurcation at the threshold defined by \( \iota_k = \eta = 0 \).
An uncoupled neuron may thus only display two different types of dynamic states.
For \( \eta \leq 0 \) the phase \( \theta_k \) ends up in a fixed point (quiescence) which we denote \quiescent.
For \( \eta > 0 \), the phase will periodically pass \( \theta = -\pi \) (spiking) which we denote \spiking.
Since both neurons are identically parametrized, two uncoupled neurons are either both quiescent (\quiescent\quiescent) or both spiking (\spiking\spiking), as shown in Fig.~\ref{fig:nonadaptive-sym-attractors}.
\begin{figure}[htp!]
  \centering
  \begin{overpic}[width=\columnwidth]{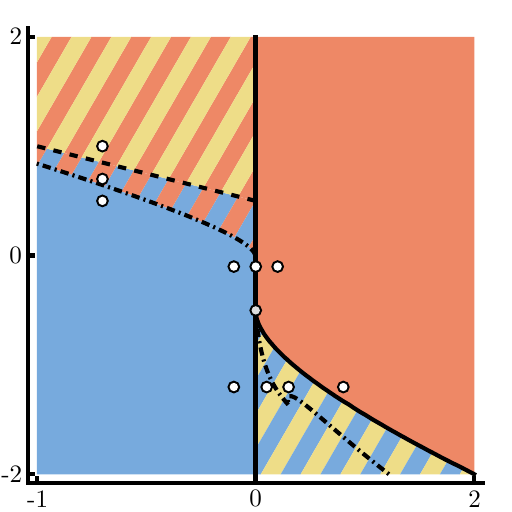}
    % Axis labels
    \put(96,4){\Large\( \eta_\sigma \)}
    \put(4,96){\Large\( \kappa \)}
    % Bifurcation markers
    \put(51.5,70){\SNone}
    \put(80,15){\SNtwo}
    \put(51.5,37){Cusp}
    % Phase portrait markers
    \put(46.5,44.5){\panellabel{sym-kappa-close0-I}}
    \put(51,44.5){\panellabel{sym-kappa-close0-II}}
    \put(55.5,44.5){\panellabel{sym-kappa-close0-III}}
    \put(46,20.5){\panellabel{sym-kappa-neg-I}}
    \put(51.5,20.5){\panellabel{sym-kappa-neg-II}}
    \put(58.5,23){\panellabel{sym-kappa-neg-III}}
    \put(69,23){\panellabel{sym-kappa-neg-IV}}
    \put(21.8,59.5){\panellabel{sym-kappa-grow-I}}
    \put(21.8,64){\panellabel{sym-kappa-grow-II}}
    \put(21.8,70.5){\panellabel{sym-kappa-grow-III}}
    % Region labels
    \put(22,35){\Large\quiescent\quiescent}
    \put(70,65){\Large\spiking\spiking}
    \put(22,80){\large\quiescentlibration\quiescentlibration/\spiking\spiking}
    \put(33,60.5){\rotatebox{-15}{\quiescent\quiescent/\spiking\spiking}}
    \put(52,10){\quiescent\quiescent/\quiescentlibration\quiescentlibration}
  \end{overpic}
  \caption{
    Stability diagram for identical neurons in the non-adaptive system \cref{eq:reduced-model} with \(a=0\) and \( \eta_1 = \eta_2 = \eta \).
    Saddle-node bifurcations are shown as solid lines, the continuum of fixed points as a dashed line and global bifurcations as dash-dotted lines.
    Regions are colored based on their attractors: \quiescent\quiescent is blue, \quiescentlibration\quiescentlibration is yellow and \spiking\spiking is orange.
    Phase portraits for the points \ref{sym-kappa-close0-I}-\ref{sym-kappa-grow-III} are shown in \cref{fig:nonadaptive-sym-phases-kappa-close0,fig:nonadaptive-sym-phases-kappa-neg,fig:nonadaptive-sym-phases-kappa-grow}
  }
  \label{fig:nonadaptive-sym-attractors}
\end{figure}
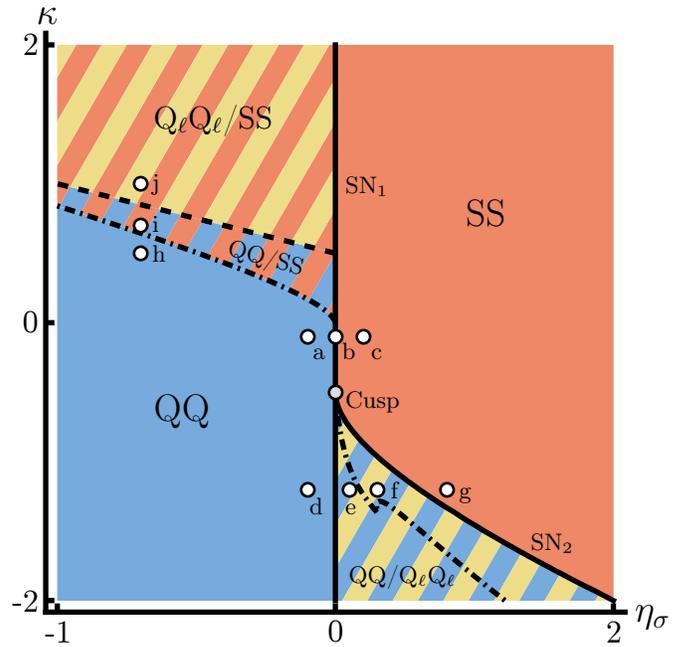

When the two neurons are coupled, \( \kappa \neq 0 \), the resulting dynamics and possible bifurcations can become more complicated.
In general, three types of attractors can be observed for identical neurons.
Both states/attractors from the uncoupled  case carry over to the coupled case, i.e., both nodes are quiescent (\quiescent\quiescent) or spiking periodically (\spiking\spiking); however, a new dynamic scenario becomes possible, where the neurons do not spike but oscillate in a small-amplitude periodic orbit (\quiescentlibration\quiescentlibration), corresponding to a libration (spiking corresponds to a rotation). 

% negative coupling 
We first focus on negative coupling strengths.
The same bifurcation scenario as for the uncoupled case ($\kappa=0$) occurs for the range of $-\half<\kappa \leq 0$:
\begin{figure}[htp!]
  \centering
  \begin{overpic}[width=.3\columnwidth]{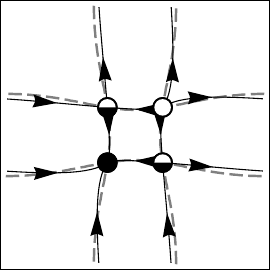}
    \put(5,86){\ref{sym-kappa-close0-I})}
  \end{overpic}
  \hfill
  \begin{overpic}[width=.3\columnwidth]{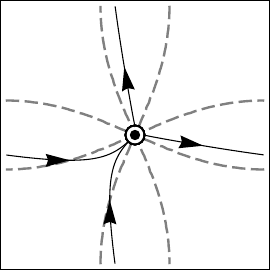}
    \put(5,86){\ref{sym-kappa-close0-II})}
  \end{overpic}
  \hfill
  \begin{overpic}[width=.3\columnwidth]{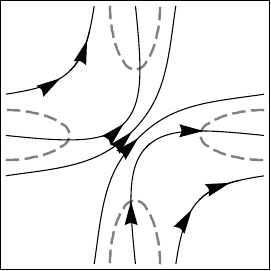}
    \put(5,86){\ref{sym-kappa-close0-III})}
  \end{overpic}
  
  \caption{
    Phase portrait for \(a=0\), \(\kappa=-0.1\)  with \( \eta = -0.1, 0, 0.1 \) and identical excitabilities, corresponding to the points \ref{sym-kappa-close0-I}-\ref{sym-kappa-close0-III} shown in \cref{fig:nonadaptive-sym-attractors}, respectively.
    (Nullclines are shown as gray dashed lines, stable/unstable/saddle fixed points are shown as filled/empty/half-filled circles.)
  }
  \label{fig:nonadaptive-sym-phases-kappa-close0}
\end{figure}
When \( \eta < 0 \) there is a stable node, unstable node and two saddles (\quiescent\quiescent, \cref{fig:nonadaptive-sym-phases-kappa-close0} \cref{sym-kappa-close0-I}).
At \( \eta = 0 \) the 4 fixed points coalesce in the origin in a saddle-node bifurcation (\SNone, \cref{fig:nonadaptive-sym-phases-kappa-close0} \cref{sym-kappa-close0-II}).
For \( \eta > 0 \), there are only spiking orbits (\spiking\spiking, \cref{fig:nonadaptive-sym-phases-kappa-close0} \cref{sym-kappa-close0-III}).
Note that these spiking states are not limit cycles, but rather are foliating the entire phase space.

As we decrease $\kappa$ further below $\half$, additional bifurcations occur.
To explain the transitions for \( \kappa < \half \) we increase \( \eta \) from negative to positive values.
\begin{figure}[htp!]
  \centering
  \begin{overpic}[width=.3\columnwidth]{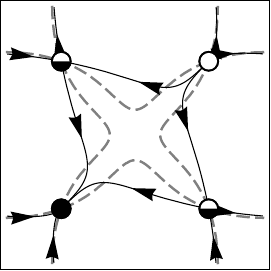}
    \put(5,86){\ref{sym-kappa-neg-I})}
  \end{overpic}
  \hspace{.05\columnwidth}
  \begin{overpic}[width=.3\columnwidth]{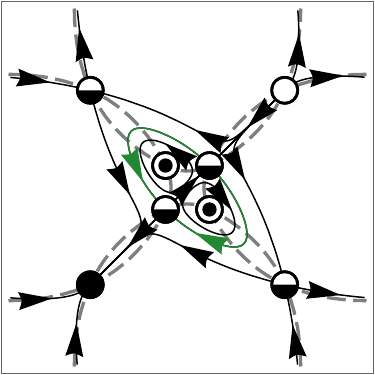}
    \put(5,86){\ref{sym-kappa-neg-II})}
  \end{overpic}
  \\[.05\columnwidth]
  \begin{overpic}[width=.3\columnwidth]{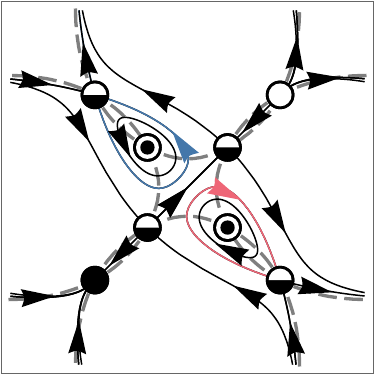}
    \put(5,86){\ref{sym-kappa-neg-III})}
  \end{overpic}
  \hspace{.05\columnwidth}
  \begin{overpic}[width=.3\columnwidth]{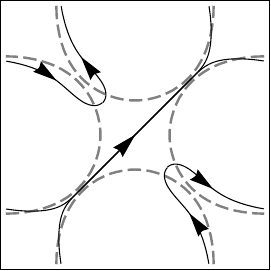}
    \put(5,86){\ref{sym-kappa-neg-IV})}
  \end{overpic}
  
  \caption{
    Phase portraits for the non-adaptive system with \(a=0\), \(\kappa=-1.2\) and \( \eta = -0.1, 0.05, 0.15, 0.4\) and identical excitabilities, corresponding to the points \ref{sym-kappa-neg-I}-\ref{sym-kappa-neg-IV} in \cref{fig:nonadaptive-sym-attractors}, respectively.
    (Nullclines are shown as gray dashed lines, stable/unstable/saddle/center fixed points are shown as filled/empty/half-filled/center-dot circles.)
  }
  \label{fig:nonadaptive-sym-phases-kappa-neg}
\end{figure}
For $\eta<0$ we have two saddles and a stable and an unstable node (\quiescent\quiescent, \cref{fig:nonadaptive-sym-phases-kappa-neg} \cref{sym-kappa-neg-I}).
As we augment $\eta$ and traverse \SNone from left to right, four additional fixed points, two saddles and two center points, are born at the origin (\cref{fig:nonadaptive-sym-phases-kappa-neg} \cref{sym-kappa-neg-II}).
We denote these as the inner fixed points, as opposed to the outer fixed points present in~\cref{sym-kappa-neg-I}.
The center points are surrounded by a foliation of librations (\quiescentlibration\quiescentlibration) bounded by heteroclinic curves connecting the two saddles (green orbits in \cref{fig:nonadaptive-sym-phases-kappa-neg} \cref{sym-kappa-neg-II}). 
Thus, this region displays bi-stability between quiescent (\quiescent\quiescent) and librating (\quiescentlibration\quiescentlibration) states.
As we increase $\eta$ further, the heteroclinic connections change their nature as the dash-dotted curve is traversed\footnote{This boundary is numerically determined by considering orbits starting near the outer saddle and determining whether the resulting orbit remains oscillatory (libration) or converges to the stable node}:
instead of connecting the saddles to  each other, they now have become homoclinic curves connecting each saddle to itself (blue and red orbits in \cref{fig:nonadaptive-sym-phases-kappa-neg} \cref{sym-kappa-neg-III}).
Note that the system still provides the same type of bi-stability (\quiescent\quiescent/\quiescentlibration\quiescentlibration).
Finally, as we further increase $\eta$, we traverse \SNtwo and all eight fixed points coalesce pairwise on the diagonals $|\theta_1|=|\theta_2|$ in a saddle-node bifurcation, leaving behind only periodic spiking orbits (\spiking\spiking, \cref{fig:nonadaptive-sym-phases-kappa-neg} \cref{sym-kappa-neg-IV}).

% positive coupling, left side of plot:
For positive coupling (\( \kappa > 0 \)), we first focus on the bifurcations in the region of negative excitability (\( \eta < 0 \)).
\begin{figure}[htp!]
  \centering
  \begin{overpic}[width=.3\columnwidth]{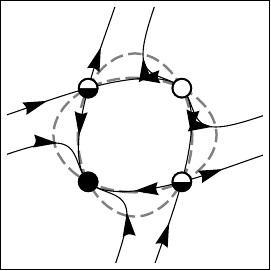}
    \put(5,86){\ref{sym-kappa-grow-I})}
  \end{overpic}
  \hfill
  \begin{overpic}[width=.3\columnwidth]{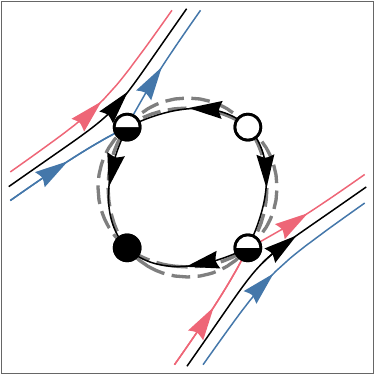}
    \put(5,86){\ref{sym-kappa-grow-II})}
  \end{overpic}
  \hfill
  \begin{overpic}[width=.3\columnwidth]{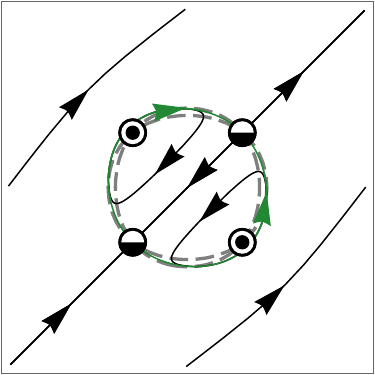}
    \put(5,86){\ref{sym-kappa-grow-III})}
  \end{overpic}
  
  \caption{
    Phase portraits for the non-adaptive system with \(a=0\), \( \eta = -0.7\) and \(\kappa = 0.5, 0.7, 1\) and identical excitabilities, corresponding to the parameter points \ref{sym-kappa-grow-I}-\ref{sym-kappa-grow-III} in \cref{fig:nonadaptive-sym-attractors}, respectively.
    Two homoclinic orbits are marked in blue and red, and two heteroclinic orbits are marked in green.
    (Nullclines are shown as gray dashed lines, stable/unstable/saddle/center fixed points are shown as filled/empty/half-filled/center-dot circles.)
    }
  \label{fig:nonadaptive-sym-phases-kappa-grow}
\end{figure}
For sufficiently low \( \kappa \), the system only displays four fixed points (stable/unstable node, two saddles) so that the only stable state corresponds to quiescence (\quiescent\quiescent, \cref{fig:nonadaptive-sym-phases-kappa-grow} \cref{sym-kappa-grow-I}).
As we increase \( \kappa \), the topology of the phase space changes substantially in a global bifurcation\footnote{The location of this bifurcation was determined numerically by sampling many initial conditions and determining whether trajectories only  converge to a quiescent fixed point (stable node \quiescent\quiescent),  or whether some trajectories correspond to rotations (\spiking\spiking) in phase space.
} (dash-dotted line in \cref{fig:nonadaptive-sym-attractors}).
The stable and unstable manifolds of the two saddles merge to form a homoclinic connection for each saddle separately (these homoclinic orbits are highlighted in red and blue in~\cref{fig:nonadaptive-sym-phases-kappa-grow} \cref{sym-kappa-grow-II}).
Between these homoclinic orbits lie a foliation of periodic spiking orbits.
Thus this region displays bi-stability between quiescent (\quiescent\quiescent) and spiking (\spiking\spiking) states.
Further increasing \( \kappa \), we cross the boundary defined by \( \nu = \eta + 2\kappa - 1 = 0 \) (dashed line in Fig.~\ref{fig:nonadaptive-sym-attractors}), where the two fixed point conditions in \cref{eq:non-adaptive-fp} become identical (i.e., nullclines are identical) and produce a continuum of fixed points.
Above this global bifurcation, on one hand, the two saddles transform into centers surrounded by a foliation of periodic librating orbits (\quiescentlibration\quiescentlibration, \cref{fig:nonadaptive-sym-phases-kappa-grow} \cref{sym-kappa-grow-III}).
Vice versa, the stable and unstable nodes have become saddles with heteroclinic connections (green orbits in \cref{fig:nonadaptive-sym-phases-kappa-grow} \cref{sym-kappa-grow-III}) separating the librations from a foliation of rotations (\spiking\spiking).
Consequently, this region is bi-stable between co-existence of \quiescentlibration\quiescentlibration and \spiking\spiking states.

% positive coupling, left to right:
Finally, traversing saddle-node bifurcation \SNone located at \( \eta = 0 \) in the regime of positive coupling, \( \kappa > 0 \), while increasing $\eta$, the fixed points coalesce at the origin in a saddle-node bifurcation, leaving behind a spiking state (\spiking\spiking) only. This scenario is similar to the ones observed for \( \half < \kappa \leq 0 \).

% ---------------------------------------
\subsection{Dynamics for non-identical neurons\texorpdfstring{ (\( \eta_1 \neq \eta_2 \))}{}}\label{sec:non-identical_dynamics}
\subsubsection{Bifurcations of fixed points and limit cycles}
We now move on to the case of non-identical excitabilities (\( \eta_1 \neq \eta_2 \)).
In what follows, we will use the parametrization \( \eta_\sigma = (\eta_1 + \eta_2) / 2 \), \( \eta_\delta = (\eta_2 - \eta_1) / 2 \).
With this parametrization, the identical case can be seen as the limiting case with \( \eta_\delta = 0 \), \( \eta = \eta_\sigma \); similarly, we will here fix an \( \eta_\delta \neq 0 \), and look at bifurcations in \(( \eta_\sigma \),\( \kappa )\) space.
Without loss of generality, we will assume that \( \eta_\delta > 0 \) (\( \eta_2 > \eta_1 \)).
Further, we will use \( \eta_\delta = 0.1 \) throughout much of our analysis. The stability diagram and phase portraits in Figs.~\ref{fig:nonadaptive-asym-attractors} to \ref{fig:nonadaptive-asym-phases-locking} summarize the possible states and bifurcations.

As for the identical case, we consider first the situation of uncoupled neurons (\( \kappa = 0 \)).
Since \( \iota_1 = \eta_1 \neq \eta_2 = \iota_2 \), the two neurons will undergo SNIC bifurcations at different thresholds: \( \eta_\sigma + \eta_\delta = \eta_2 = \iota_2 = 0 \) and \( \eta_\sigma - \eta_\delta = \eta_1 = \iota_1 = 0 \) for neurons 2 and 1, respectively.
The non-identical uncoupled system can thus have 3 states:
When \( \eta_\sigma < -\eta_\delta \) both neurons are quiescent (\quiescent\quiescent), when \( -\eta_\delta < \eta_\sigma < \eta_\delta \), neuron 1 is quiescent while neuron two is spiking (\quiescent\spiking) and when \( \eta_\delta < \eta_\sigma \) both neurons are spiking (\spiking\spiking).
These transitions are shown in~\cref{fig:nonadaptive-asym-attractors}.
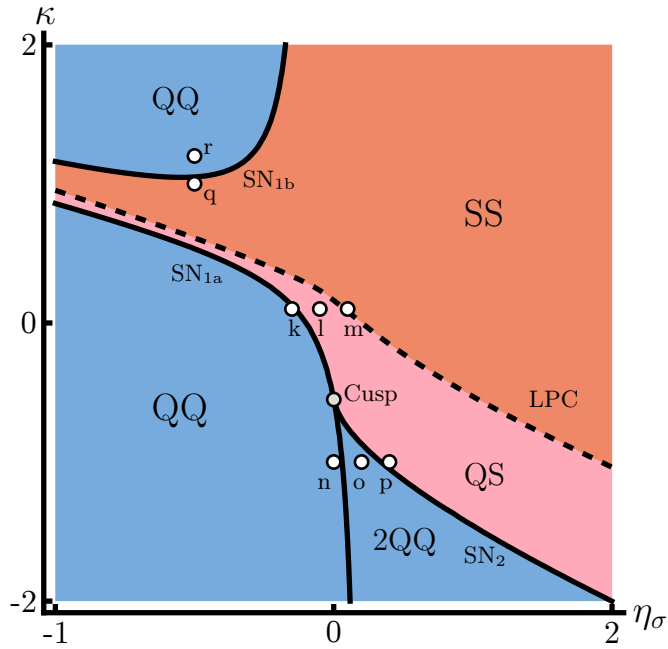
\begin{figure}[htp!]
  \centering
  \begin{overpic}[width=\columnwidth]{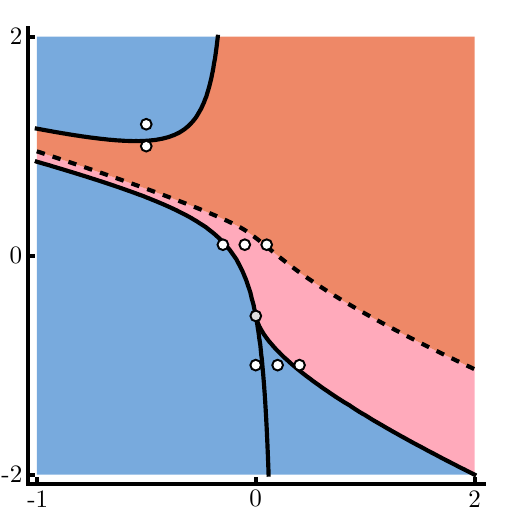}
    % Axis labels
    \put(96,4){\Large\( \eta_\sigma \)}
    \put(4,96){\Large\( \kappa \)}
    % Bifurcation markers
    \put(36,71){\SNoneb} % Upper SN1
    \put(25,56.5){\SNonea} % Lower SN1
    \put(70,13){\SNtwo}
    \put(80,37){LPC}
    \put(51.5,38){Cusp}
    % Phase portrait markers
    \put(43,48){\panellabel{asym-close0-I}}
    \put(47.5,48){\panellabel{asym-close0-II}}
    \put(51.5,48){\panellabel{asym-close0-III}}
    \put(47.5,24.5){\panellabel{asym-neg-I}}
    \put(53,24.5){\panellabel{asym-neg-II}}
    \put(57,24.5){\panellabel{asym-neg-III}}
    \put(30,69){\panellabel{asym-pos-I}}
    \put(30,76){\panellabel{asym-pos-II}}
    % Region labels
    \put(22,35){\Large\quiescent\quiescent}
    \put(70,65){\Large\spiking\spiking}
    \put(22,83){\large\quiescent\quiescent}
    \put(70,25){\large\quiescent\spiking}
    \put(56,15){\large2\quiescent\quiescent}
  \end{overpic}
  
  \caption{
    Stability diagram for non-identical neurons in the non-adaptive system \cref{eq:reduced-model} with \(a=0\) and \( \eta_\delta = 0.1 \).
    Saddle-node bifurcations are shown as solid lines, and limit points of (LPC) cycles as a dashed line.
    Regions are colored based on their attractors: \quiescent\quiescent is blue, \quiescent\spiking is pink and \spiking\spiking is orange.
    Phase portraits for the points \ref{asym-close0-I}-\ref{asym-pos-II} are shown in \cref{fig:nonadaptive-asym-phases-kappa-close0,fig:nonadaptive-asym-phases-kappa-neg,fig:nonadaptive-asym-phases-kappa-pos}
  }
  \label{fig:nonadaptive-asym-attractors}
\end{figure}

For two coupled neurons (\( \kappa \neq 0 \)), the input to neuron 1 \( \iota_1 \) will vary with time , and so the \quiescent\spiking state does --- strictly speaking --- not exist.
Instead, we can observe the state \quiescentlibration\spiking, where neuron 1 is \enquote{chasing} a moving stable phase.
However, looking topologically at the phase space of the system, the \quiescentlibration\spiking state is not distinct from the \quiescent\spiking state, and we will therefore treat them as the same type of state.

Next, we focus on coupling \( \kappa \) close to 0, specifically less than the minimum value of the upper component of \SNoneb (\( \kappa < \abs{\eta_\delta} + \sqrt{2 \abs{\eta_\delta}} + \half \)) and greater than the value at the cusp point where \SNonea and \SNtwo meet (\( \kappa > -(1 + \abs{\eta_\delta}) / 2 \)).
In this region, the same bifurcations occur when varying \( \eta_\sigma \) as in the uncoupled case.
However, since the neurons now are coupled, we must consider the bifurcations of individual neurons as bifurcations of the entire system.
\begin{figure}[htp!]
  \centering
  \begin{overpic}[width=.3\columnwidth]{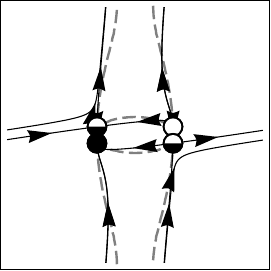}
    \put(5,86){\ref{asym-close0-I})}
  \end{overpic}
  \hfill
  \begin{overpic}[width=.3\columnwidth]{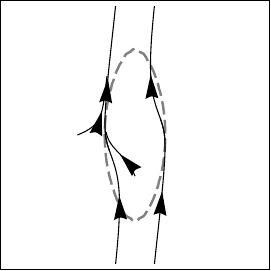}
    \put(5,86){\ref{asym-close0-II})}
  \end{overpic}
  \hfill
  \begin{overpic}[width=.3\columnwidth]{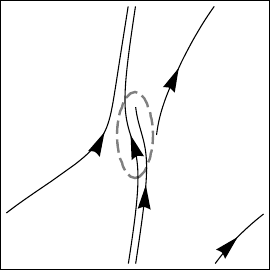}
    \put(5,86){\ref{asym-close0-III})}
  \end{overpic}
  
  \caption{
    Phase portraits for the non-adaptive system with \(a=0\), \(\kappa=0.1\), \( \eta_\delta = 0.1 \), and \(\eta_\sigma = -0.15,-0.05, 0.05\), corresponding to the parameter points \ref{asym-close0-I}-\ref{asym-close0-III} in \cref{fig:nonadaptive-asym-attractors}, respectively.
    (Nullclines are shown as gray dashed lines, stable/unstable/saddle fixed points are shown as filled/empty/half-filled circles.)
  }
  \label{fig:nonadaptive-asym-phases-kappa-close0}
\end{figure}
For small \( \eta_\sigma \), we have two saddles and a stable and an unstable node (\quiescent\quiescent, \cref{fig:nonadaptive-asym-phases-kappa-close0} \cref{asym-close0-I}).
As \( \eta_\sigma \) is increased and passes the lower component of \SNone, the two pairs of left and right fixed points respectively coalesce in a saddle-node bifurcation on \( \theta_2 = 0 \).
What is left is a stable and an unstable limit cycle that crosses \( \theta_2 = \pi \) (corresponding to a spike), but where only \( \theta_1 \) librates (\quiescent\spiking, \cref{fig:nonadaptive-asym-phases-kappa-close0} \cref{asym-close0-II}).
Further increasing \( \eta_\sigma \), the two limit cycles coalesce in a limit point of cycles (LPC), and for larger \( \eta_\sigma \) there are only periodic orbits where both neurons spike (\spiking\spiking, \cref{fig:nonadaptive-asym-phases-kappa-close0} \cref{asym-close0-III}).

Moving our attention to the region below the cusp point where \( \kappa < -(1 + \abs{\eta_\delta}) / 2 \), starting at small \( \eta_\sigma \) there are 4 fixed points (stable/unstable, 2 saddles) (\quiescent\quiescent, \cref{fig:nonadaptive-asym-phases-kappa-neg} \cref{asym-neg-I}).
\begin{figure}[htp!]
  \centering
  \begin{overpic}[width=.3\columnwidth]{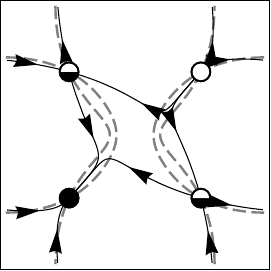}
    \put(5,86){\ref{asym-neg-I})}
  \end{overpic}
  \hfill
  \begin{overpic}[width=.3\columnwidth]{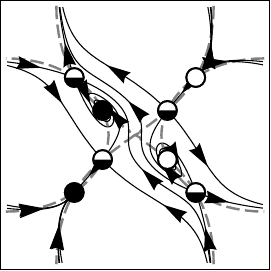}
    \put(5,86){\ref{asym-neg-II})}
  \end{overpic}
  \hfill
  \begin{overpic}[width=.3\columnwidth]{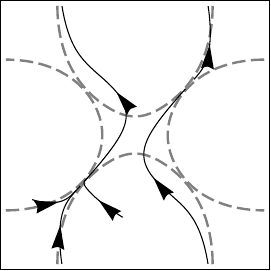}
    \put(5,86){\ref{asym-neg-III})}
  \end{overpic}
  
  \caption{
    Phase portraits for the non-adaptive system with \(a=0\), \(\kappa = -1\) and \( \eta_\delta = 0.1 \), and \(\eta_\sigma = 0, 0.1, 0.2\), 
    corresponding to the parameter points \ref{asym-neg-I}-\ref{asym-neg-III} in \cref{fig:nonadaptive-asym-attractors}, respectively.
    (Nullclines are shown as gray dashed lines, stable/unstable/saddle fixed points are shown as filled/empty/half-filled circles.)
  }
  \label{fig:nonadaptive-asym-phases-kappa-neg}
\end{figure}
Increasing \( \eta_\sigma \) across \SNone, four more fixed points (stable/unstable spirals, 2 saddles) which we denote as \enquote{inner fixed points} are created in a saddle-node bifurcation on \( \theta_2 = 0 \) (\cref{fig:nonadaptive-asym-phases-kappa-neg} \cref{asym-neg-II}).
Thus, these region exhibits bi-stability with two fixed points corresponding to quiescence, and so we denote the region with (2\quiescent\quiescent). 
Further increasing \( \eta_\sigma \) across \SNtwo, the inner and outer fixed points coalesce in each quadrant in a saddle-node bifurcation, and leaves a stable and an unstable limit cycle the only spike in neuron 2 (\quiescent\spiking, \cref{fig:nonadaptive-asym-phases-kappa-neg} \cref{asym-neg-III}).
Finally, as in the case for \( \kappa \) closer to 0, for large \( \eta_\sigma \) the two limit cycles coalesce and disappear in an LPC; thus, only states where both neurons spike (\spiking\spiking) are left.

Considering larger \( \kappa>0 \), for \( \eta_\sigma < \abs{\eta_\delta} \), there is a further saddle-node bifurcation \SNoneb bordering the \spiking\spiking-region.
\begin{figure}[htp!]
  \centering
  \begin{overpic}[width=.3\columnwidth]{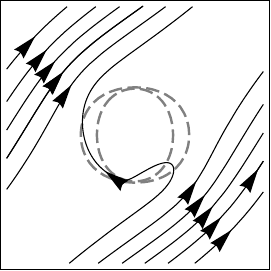}
    \put(85,86){\ref{asym-pos-I})}
  \end{overpic}
  \hspace{.05\columnwidth}
  \begin{overpic}[width=.3\columnwidth]{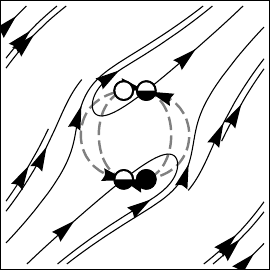}
    \put(30,86){\ref{asym-pos-II})}
  \end{overpic}

  \caption{
    Phase portraits for the non-adaptive system with \(a=0\), \(\eta_\sigma=-0.5\) and \( \eta_\delta = 0.1 \) with \(\kappa = 1, 1.1\), corresponding to the parameter points \ref{asym-pos-I} and \ref{asym-pos-II} in \cref{fig:nonadaptive-asym-attractors}, respectively.
    (Nullclines are shown as gray dashed lines, stable/unstable/saddle fixed points are shown as filled/empty/half-filled circles.)
  }
  \label{fig:nonadaptive-asym-phases-kappa-pos}
\end{figure}
Crossing \SNoneb  from inside the \spiking\spiking-region by increasing \( \kappa \), 4 fixed points (stable/unstable node, 2 saddles) appear in a saddle-node bifurcation on \( \theta_1 = 0 \), see \cref{fig:nonadaptive-asym-phases-kappa-pos} panels \ref{asym-pos-I}) and \ref{asym-pos-II}).

\subsubsection{Mode-locking}
The \spiking\spiking region (orange shading in \cref{fig:nonadaptive-asym-attractors}) features mode-locking regions, where the spiking activity of the two neurons locks into ratios of fixed frequency.
These mode-locking regions are defined by a constant average ratio, $R$, of spikes $n_1(t)$ and $n_2(t)$ occurring in neurons 1 and 2 over a time period $t$. Formally, we express this ratio as
\begin{equation}
    R = \lim_{t \to \infty} \frac{n_1(t)}{n_2(t)} =  \lim_{t \to \infty} \frac{\int_0^t \delta(\theta_1(s) - \pi) \dl s}{\int_0^t \delta(\theta_2(s) - \pi) \dl s},
\end{equation}
where \( \delta \) is the Dirac distribution on the circle \( \torus \).
This ratio $R$ equals the inverse fraction of the periods \( T_k \) of the two nodes.
Thus, orbits can only be periodic (and possibly be locked with the corresponding frequency mode) whenever $R$ is a rational number, \( R = p:q\) with \( p,q\in\naturals \).

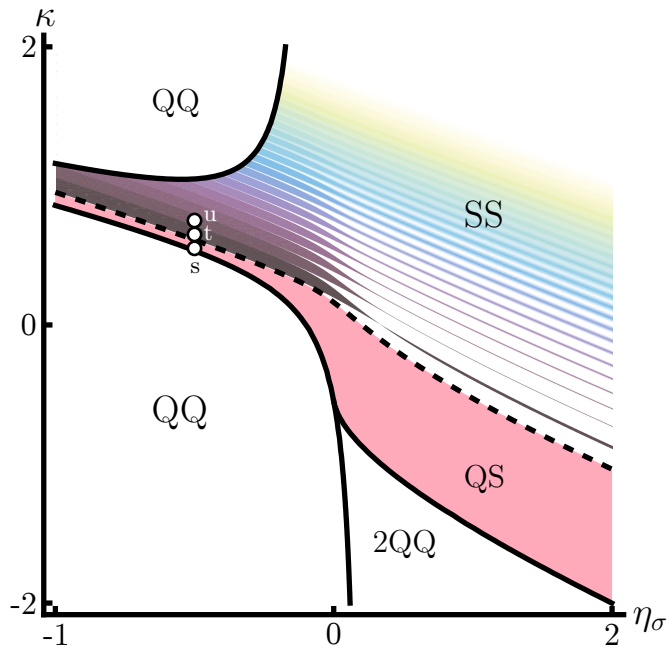
\begin{figure}[htp!]
  \centering
  \begin{overpic}[width=\columnwidth]{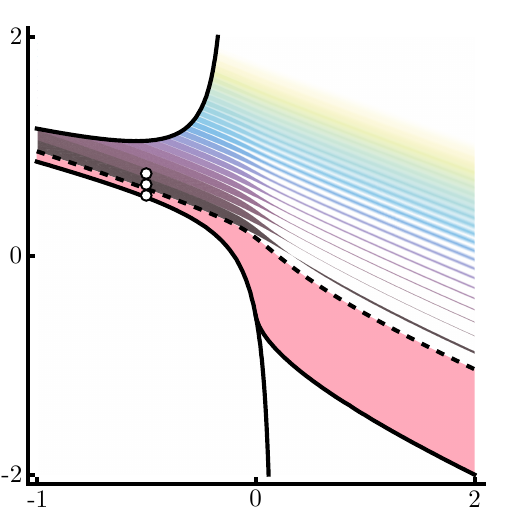}
    % Axis labels
    \put(96,4){\Large\( \eta_\sigma \)}
    \put(4,96){\Large\( \kappa \)}
    % Phase portrait markers
    \put(28,58){\panellabel{asym-locking-I}}
    \put(30,63){\textcolor{white}{\panellabel{asym-locking-II}}}
    \put(30,66){\textcolor{white}{\panellabel{asym-locking-III}}}
    % Region labels
    \put(22,35){\Large\quiescent\quiescent}
    \put(70,65){\Large\spiking\spiking}
    \put(22,83){\large\quiescent\quiescent}
    \put(70,25){\large\quiescent\spiking}
    \put(56,15){\large2\quiescent\quiescent}
  \end{overpic}
  \caption{
    Stability diagram for the non-adaptive system \cref{eq:reduced-model} with \(a=0\), \( \eta_\delta = 0.1 \).
    The \spiking\spiking region exhibits mode-locking for the spiking frequency. We only highlighted a selection of the frequency ratios, \( R = n/(n+1),\, n \in \naturals \). For a more complete view on the underlying structure, see the staircase shown in Fig.~\ref{fig:nonadaptive-asym-ratio-bifurcation}. 
    Furthermore the ratio $R$ was estimated numerically and cannot resolve every mode-locking ratio.
    Phase portraits for the points \ref{asym-locking-I}-\ref{asym-locking-III} are shown in \cref{fig:nonadaptive-asym-phases-locking}.
  }
  \label{fig:nonadaptive-asym-ratio}
\end{figure}

The mode-locking regions form Arnol'd tongues that emanate from the limit of zero coupling with $\kappa=0$.
The period of uncoupled neurons $k=1,2$ with \( \kappa = 0 \), $T_k=\pi/\sqrt{\eta_k}$ , is easily obtained analytically~\cite{gutkin2022theta},  and we obtain
\begin{equation}\label{eq:spikingratio}
    R = \frac{T_2}{T_1} =  \frac{\sqrt{\eta_1}}{\sqrt{\eta_2}} = \frac{\sqrt{\eta_\sigma - \eta_\delta}}{\sqrt{\eta_\sigma + \eta_\delta}}.
\end{equation}
Thus, Arnol'd tongues have their origin at $\kappa=0$ with excitability values $\eta_1$ and $\eta_2$,  consistent with the condition \(R \in \rationals\) according to \eqref{eq:spikingratio}.
 
However, to classify the neuron's spiking activity by their spiking ratio \( R \) for non-zero coupling (\( \kappa \neq 0 \)), we need to approximate the spiking ratios by numerically integrating trajectories (limit cycles) and measuring the ratio \( n_1(T) : n_2(T) \) for a large time \( T \).
To detect the mode-locking regions shown in \cref{fig:nonadaptive-asym-ratio}, we estimated spiking ratios that lie within a tolerance of 0.003. A variety of ratios occur, but we specifically chose to only highlight ratios of the form \( n : (n + 1)\) where \( n \in \naturals \), since those ratios correspond to the widest regions and, thus, better reveal the parameter dependence in the stability diagram.

For negative and smaller positive values of \( \eta_\sigma \), the \( n : (n + 1) \) regions cover almost all of the parameter-space, while they almost abruptly become much thinner around a certain small positive value of $\eta_\sigma$.
In \cref{fig:nonadaptive-asym-phases-locking}, we show the phase portraits for the modes \( 0 \), \( 1 : 2 \) and \( 2 : 3 \).
\begin{figure}[htp!]
  \centering
  \begin{overpic}[width=.3\columnwidth]{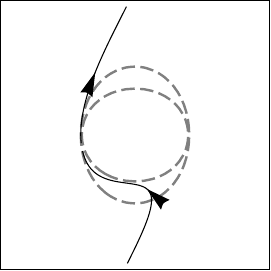}
    \put(5,86){\ref{asym-locking-I})}
  \end{overpic}
  \hfill
  \begin{overpic}[width=.3\columnwidth]{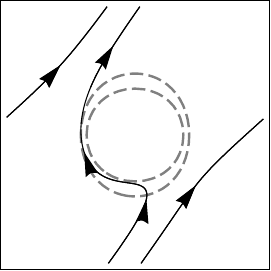}
    \put(5,86){\ref{asym-locking-II})}
  \end{overpic}
  \hfill
  \begin{overpic}[width=.3\columnwidth]{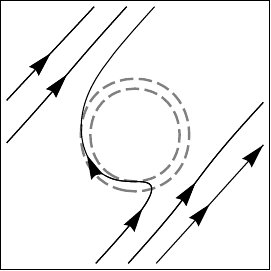}
    \put(5,86){\ref{asym-locking-III})}
  \end{overpic}
  
  \caption{
    Phase portraits for the non-adaptive system with \(a=0\), \(\eta_\sigma=-0.5\) and \( \eta_\delta = 0.1 \) with \(\kappa = 0.55, 0.65, 0.75\), corresponding to the parameter points \ref{asym-locking-I}-\ref{asym-locking-III} in \cref{fig:nonadaptive-asym-ratio}, respectively.
    (Nullclines are shown as gray dashed lines, stable/unstable/saddle fixed points are shown as filled/empty/half-filled circles.)
  }
  \label{fig:nonadaptive-asym-phases-locking}
\end{figure}
Note that the mode with \(  0 : 1 = 0\) corresponds to the \quiescent\spiking state we already discussed further above.

\begin{figure}[htp!]
  \centering
  \begin{overpic}[width=\columnwidth]{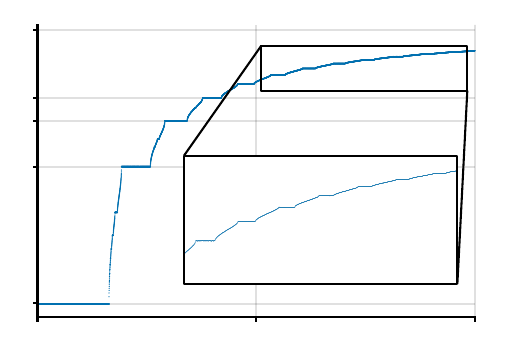}
    % Axis labels
    \put(5,63){\Large\( R \)}
    \put(94,4){\Large\( \kappa \)}
    % Tick labels
    %   y ticks
    \put(-1,59){\( 1:1 \)}
    \put(-1,46.5){\( 3:4 \)}
    \put(-1,41.7){\( 2:3 \)}
    \put(-1,33){\( 1:2 \)}
    \put(4,6.4){\( 0 \)}
    %   x ticks
    \put(6.5,0){\( 0 \)}
    \put(48,0){\( 0.5 \)}
    \put(92,0){\( 1 \)}
  \end{overpic}
  \caption{   
    Numerical estimation of frequency ratios of spikes $R$ for varying $\kappa$ with $\eta_\sigma=0$.
    Plateaus of frequency-locking corresponding to the colored regions in \cref{fig:nonadaptive-asym-ratio} can be seen, and form a devil's staircase. 
  }
  \label{fig:nonadaptive-asym-ratio-bifurcation}
\end{figure}
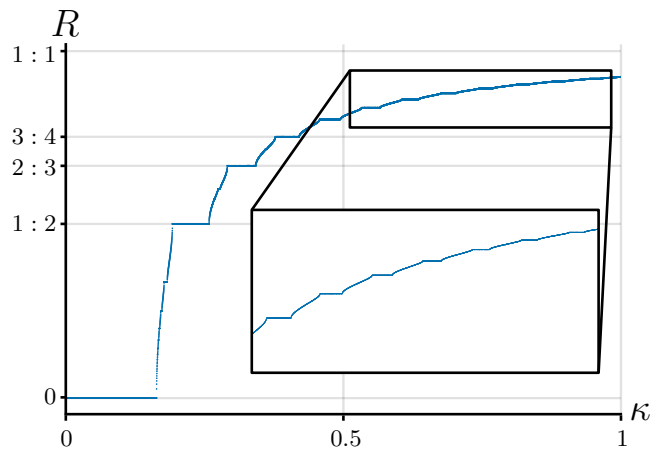
Measuring the spike ratio with \( \eta_\sigma \) fixed while increasing the coupling strength \(\kappa\), 
\cref{fig:nonadaptive-asym-ratio-bifurcation} reveals a Devil's staircase formed by the mode-locking regions.
The \( n : (n + 1) \) modes are visible as steps in the staircase, but we also observe many other locked modes which appear smaller in size.

\subsubsection{Comparison of dynamics occurring for  identical versus non-identical neurons}\label{sec:nonadaptive-comparison}
While the non-adaptive system features three control parameters (\( \eta_\sigma \), \( \eta_\delta \) and \( \kappa \)), the special case of identical excitabilities is represented by the parameter plane with \( \eta_\delta = 0 \).
We have seen that a variety of dynamic structures observed for that plane, including librations around center points and bi-stability, disappear in the non-identical case. These structures are \emph{not structurally robust} towards perturbations such as breaking the system symmetry of $\eta_\delta=0$, and the only structures that are left are fixed points, limit cycles and saddle-node bifurcations.
Notably, mode-locking regions are entirely absent for identical neurons. Although, in some sense one might be tempted to say that the spiking orbits in the \spiking\spiking region are locked; but this interpretation is problematic since the center orbits are not attractive and thus in reality provide no locking mechanism.

% ==========================

\section{Analysis for adaptive coupling\texorpdfstring{ (\( a > 0 \))}{}}
We now consider the dynamics for adaptive coupling  with \( a > 0 \). For simplicity we set the baseline coupling \( b = 0 \) and the time scale of adaptation to \( \varepsilon = 0.01 \).
Since the dynamics for identical excitabilities in the non-adaptive case turned out to be structurally non-robust (\cref{sec:nonadaptive-comparison}), we focus on non-identical excitabilities with \( \eta_\delta = (\eta_2 - \eta_1) / 2 = 0.1 \).
The remaining control parameters are then the average excitability \( \eta_\sigma = (\eta_1 + \eta_2) / 2 \) and the adaptivity parameter \( a \).

% ----------------------------------------

\subsection{Regimes of quiescent neurons (\quiescent\quiescent) and of one spiking neuron (\quiescent\spiking)}
We first focus on the two states of two quiescent neurons (\quiescent\quiescent) and of one single spiking neuron (\quiescent\spiking). The corresponding stability regions are shown in \cref{fig:coevolutionary-asym-attractors} with distinct gray shadings.
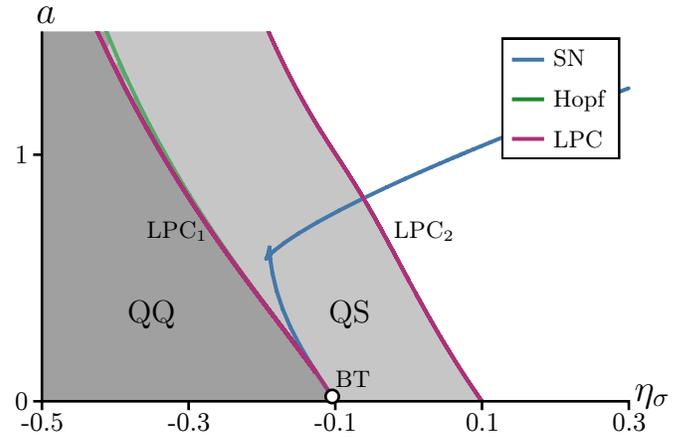
\begin{figure}[htp!]
  \centering
  \begin{overpic}[width=\columnwidth]{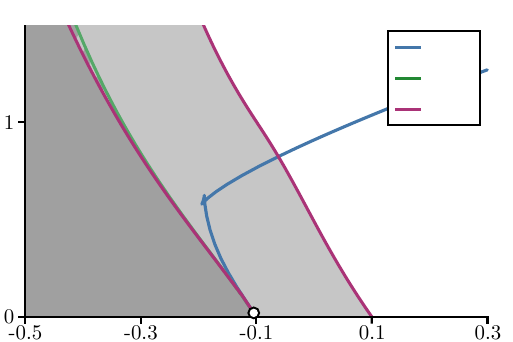}
    % Axis labels
    \put(96,5){\Large\( \eta_\sigma \)}
    \put(4,63){\Large\( a \)}
    % Co-dimension 2 labels
    \put(50,7){BT}
    % Line labels
    \put(21,30){\LPCone}
    \put(59,30){\LPCtwo}
    % Region labels
    \put(18,17){\large\quiescent\quiescent}
    \put(49,17){\large\quiescent\spiking}
    % Legend text
    \put(83.5,56.5){SN}
    \put(83.5,50.25){Hopf}
    \put(83.5,44){LPC}
  \end{overpic}
  \caption{
    Stability diagram for the model with adaptation. Parameters are  \( \varepsilon = 0.01 \), \( b = 0 \), and \( \eta_\delta = 0.1 \).
  }
  \label{fig:coevolutionary-asym-attractors}
\end{figure}
The non-adaptive case with \( a = 0 \) corresponds in the asymptotic time limit to the uncoupled system ($\kappa=0$) studied in \cref{sec:non-identical_dynamics}; namely,
the system undergoes the transition called \SNone from \quiescent\quiescent to  \quiescent\spiking in a SNIC bifurcation at \( \eta_\sigma = - \abs{\eta_\delta} = -0.1 \).
However, note that two changes occur at the SNIC bifurcation:
firstly, increasing \( \eta_\sigma \), two fixed points coalesce and disappear at the SNIC;
secondly, while the associated fixed points previously lived on an invariant and stable curve in phase space, this curve transforms into a limit cycle.

For non-zero adaptivity \( a \neq 0 \), the SNIC bifurcation disappears and instead a region of bi-stability between \quiescent\quiescent and \quiescent\spiking states is formed, as shown in the bifurcation diagram in \cref{fig:coevolutionary-asym-bifurcation-diagram-kappa-minmax} for a small adaptivity value $a=0.01$.  
\begin{figure}[htp!]
  \begin{overpic}[width=\columnwidth]{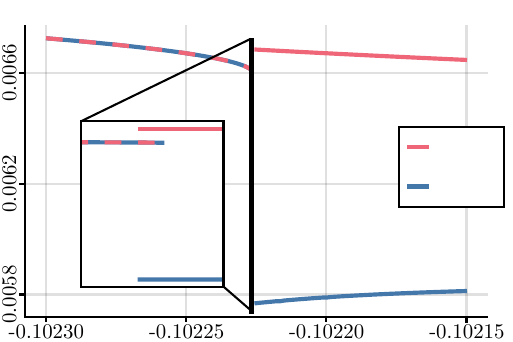}
    % Axis labels
    \put(96,6){\Large\( \eta_\sigma \)}
    \put(3,63){\Large\( \kappa \)}
    % Legend
    \put(84.5,36.8){\(\max_t(\kappa)\)}
    \put(84.5,29.2){\(\min_t(\kappa)\)}
    % States
    \put(10,51){\large \quiescent\quiescent}
    \put(54,34){\large \quiescent\spiking}
  \end{overpic}
  \caption{
    Bifurcation diagram displaying minima and maxima detected in time, i.e., \( \min_t(\kappa) \) and \( \max_t(\kappa) \)  for \( a = 0.01 \) (\( \eta_\delta = 0.1 \), \( \varepsilon = 0.01 \) and \( b = 0 \)) when performing quasi-continuation from both the left and right.
    The maxima are red lines and the minima blue lines.
  }
  \label{fig:coevolutionary-asym-bifurcation-diagram-kappa-minmax}
\end{figure}
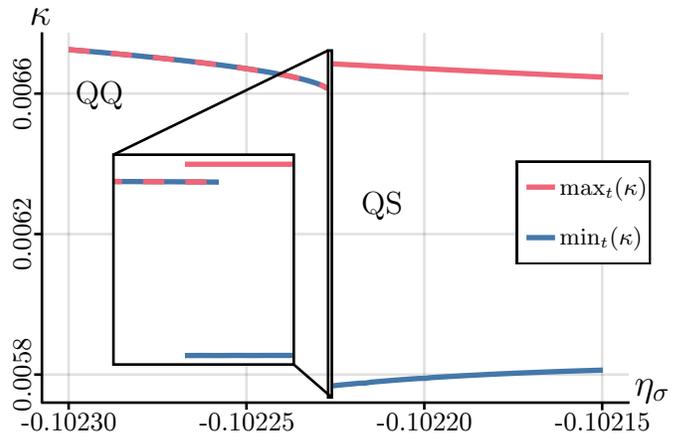
The two topological changes in phase space now occur in different bifurcations; the fixed points coalesce in a saddle-node bifurcation (SN, blue line in \cref{fig:coevolutionary-asym-attractors,fig:coevolutionary-asym-attractors-small-a}), while the limit cycle is born from a limit point of cycles (\LPCone, purple line in \cref{fig:coevolutionary-asym-attractors,fig:coevolutionary-asym-attractors-small-a}). 
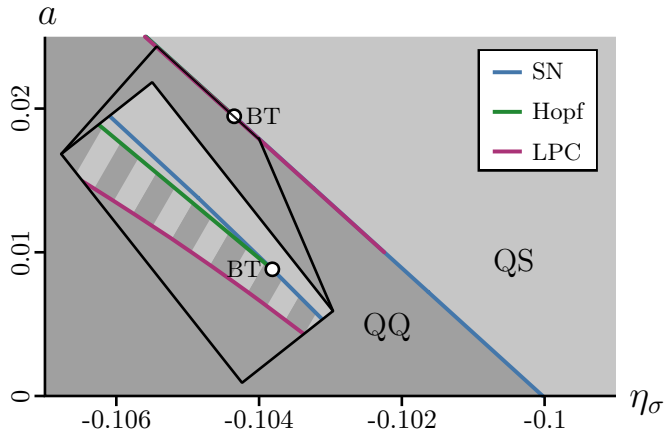
\begin{figure}[htp!]
  \begin{overpic}[width=\columnwidth]{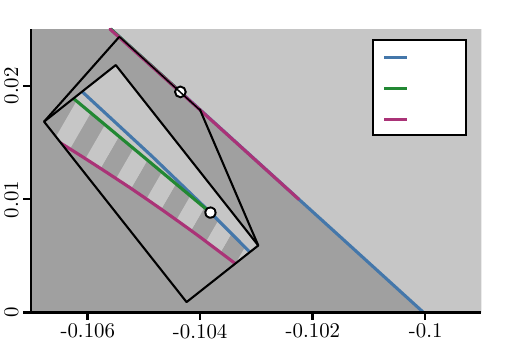}
    % Axis labels
    \put(96,4){\Large\( \eta_\sigma \)}
    \put(5,63){\Large\( a \)}
    % Codimension 2 labels
    \put(37,47.5){BT}
    \put(34,24){BT}
    % Region labels
    \put(55,15){\large\quiescent\quiescent}
    \put(75,25){\large\quiescent\spiking}
    % Legend text
    \put(81,54.5){SN}
    \put(81,48.25){Hopf}
    \put(81,42){LPC}
  \end{overpic}
  \caption{Stability diagram of attractors of the co-evolutionary system for small \( a \), \( \eta_\delta = 0.1 \).}
  \label{fig:coevolutionary-asym-attractors-small-a}
\end{figure}

For \( a \approx 0.0195 \), a Bogdanov--Takens bifurcation BT (see \cref{fig:coevolutionary-asym-attractors-small-a}) occurs on the saddle-node curve SN.
Above the BT point, the quiescent (\quiescent\quiescent) state loses stability in a subcritical Hopf bifurcation H (green line in \cref{fig:coevolutionary-asym-attractors,fig:coevolutionary-asym-attractors-small-a}); the resulting unstable \quiescent\quiescent state is annihilated in SN.
Above the BT, the bi-stable region between \quiescent\quiescent and \quiescent\spiking is therefore bounded by the \LPCone and H curves.
Thus, for \( a \)-values above BT, the saddle-node bifurcation SN no longer influences the stability regions, as it merely annihilates a fixed rendered unstable by the Hopf bifurcation (The associated saddle-node curves undergo two cusp bifurcations which we for simplicity did not highlight.).
Note that --- even for larger \( a \) --- the bi-stable region is very thin and is therefore shown enlarged in the inset in  \cref{fig:coevolutionary-asym-attractors-small-a}.

Lastly, we observe that the single neuron spiking state \quiescent\spiking is annihilated for large \( \eta_\sigma \) in another limit point of cycles (\LPCtwo, \cref{fig:nonadaptive-asym-attractors}).
Note that this boundary only denotes where the \quiescent\spiking ends; the \spiking\spiking states are sometimes bi-stable with the \quiescent\spiking states. We discuss the nature and parameter range for \spiking\spiking states in the following section.

%----------------------
\subsection{Regime of spiking neurons (\spiking\spiking)}
We now turn our attention to the regime in which two neurons are spiking (\spiking\spiking). Since the \spiking\spiking states display mode-locking (\cref{sec:non-identical_dynamics}), the interesting question arises of how these modes are affected by adaptive dynamics when $a>0$.
As for the non-adaptive case, we are concerned with the spiking ratio $R$ of \spiking\spiking states, in particular periodic oscillations with \( R \in \rationals \), for which mode-locking is present.

As we vary the adaptivity with (\( a > 0 \)), the Arnol'd tongues associated to each mode-locking regions become apparent, see \cref{fig:coevolutionary-asym-ratio-many}(a).
\begin{figure}[htp!]
  \centering
  \begin{overpic}[width=\columnwidth]{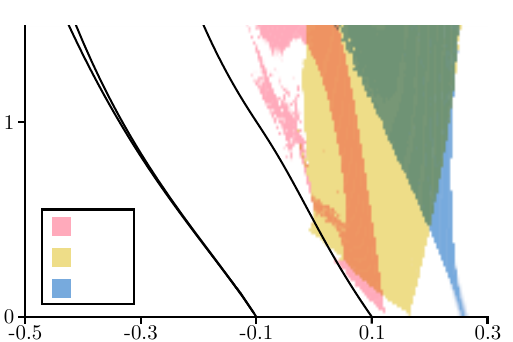}
    \put(-1,60){a)}
    % Axis labels
    \put(96,5){\Large\( \eta_\sigma \)}
    \put(4,63){\Large\( a \)}
    % Legend text
    \put(15,21.4){\( n \):\( 3n \)}
    \put(15,15.2){\( n \):\( 2n \)}
    \put(15,9){\( 2n \):\( 3n \)}
      % Line labels
    \put(9,50){\LPCone}
    \put(42,58){\LPCtwo}
    % Region labels
    \put(30,10){\large\quiescent\quiescent}
    \put(54,10){\large\quiescent\spiking}
  \end{overpic}
  \\ \medskip
  \begin{overpic}[width=\columnwidth]{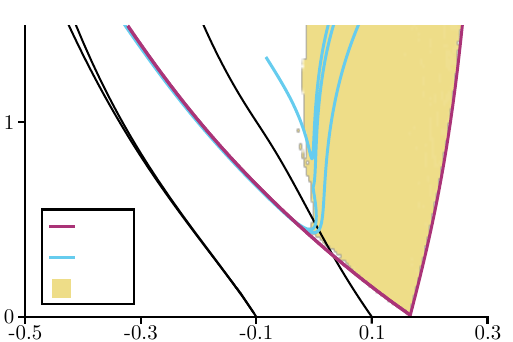}
    \put(-1,60){b)}
    % guide line
    \put(5,42.8){\tikz \draw[dashed,gray] (0,0) -- (7.8cm,0);}
    \put(5,111.7){\tikz \draw[dashed,gray] (0,0) -- (7.8cm,0);}
    % Axis labels
    \put(96,5){\Large\( \eta_\sigma \)}
    \put(4,63){\Large\( a \)}
    % Legend text
    \put(16,21.4){LPC}
    \put(16,15.2){PD}
    \put(16,9){\( n \):\( 2n \)}
    % Line labels
    \put(9,50){\LPCone}
    \put(42,58){\LPCtwo}
    % Region labels
    \put(30,10){\large\quiescent\quiescent}
    \put(54,10){\large\quiescent\spiking}
  \end{overpic}
  \caption{
    Stability diagrams for the co-evolutionary adaptive model with \( \varepsilon = 0.01 \),  \( b = 0 \), and \( \eta_\delta = 0.1 \). Different aspects of interest are highlighted in panels (a) and (b).
    {(a)} \spiking\spiking states are limit cycles which mode-lock with various ratios \( R \). Regions corresponding to ratios \( 1:3 \), \( 1:2 \) and \( 2:3 \) are highlighted.
    Selected bifurcations curves shown in \cref{fig:coevolutionary-asym-attractors} are displayed for guidance (black lines).
% % 
    {(b)} Bifurcations for the mode-locking region with \( R = n:2n =  1:2 \) (highlighted in yellow).
    Bifurcations related to limit cycles with this ratio $R$ are shown as colored curves:
    limit point of cycles (LPC) bounding the $n:2n$ limit cycles are shown in purple; associated period doubling (PD) bifurcations are shown in light blue.
    The dashed gray lines indicate the parameter sweep for the diagrams in \cref{fig:coevolutionary-asym-bifurcation-diagram-cascade}.
  }
  \label{fig:coevolutionary-asym-ratio-many}
\end{figure}
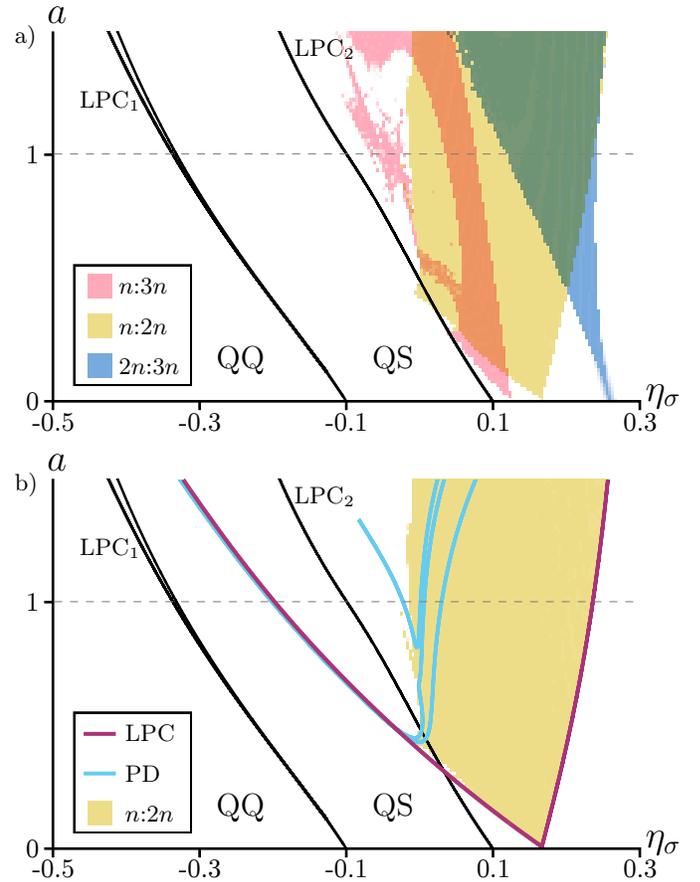
Interestingly, we observe that the tongues are widening as the adaptivity is increased---this indicates a relationship between the adaptivity $a$, and the coupling strength, $\kappa$, which typically is associated with the widening off mode-locking regions and their associated Arnol'd tongues. Furthermore, multi-stability is possible between the different \spiking\spiking mode-locking regimes, as well as between the \quiescent\spiking state and certain \spiking\spiking modes.
The data in \cref{fig:coevolutionary-asym-ratio-many}(a) was generated by estimating \( R \) numerically for initial conditions on a grid with even spacing (10, 10, and 7 grid points along the coordinates $\theta_1$, $\theta_2$, $\kappa$, respectively), and therefore the regions identified are not exhaustive.

To further understand the Arnol'd tongues, we focus on the mode with \( R = 1:2 \).
Numerical continuation with MatCont reveals that the tongue is bounded on both sides by LPCs (purple lines in \cref{fig:coevolutionary-asym-ratio-many}(b)).
For larger adaptivity \( a \), numerical continuation also shows that the \spiking\spiking limit cycle may undergo a series of period doubling bifurcations (light blue lines in \cref{fig:coevolutionary-asym-ratio-many}(b)).
Shown are the first three bifurcations of a series of period doubling bifurcations,  forming a cascade that ultimately leads to chaos as  decrease $\eta_\sigma$, as seen in \cref{fig:coevolutionary-asym-bifurcation-diagram-cascade}(b). Throughout this cascade and far into the chaotic region, the ratio $R$ remains constant (see \cref{fig:coevolutionary-asym-bifurcation-diagram-cascade}(a); however, at some critical value of $\eta_\sigma$  (indicated by an arrow), there appears to be a crisis of (chaotic) attractors, which possibly even are associated to other resonant modes. As a consequence, the fixed spiking ratio $R$ is compromised and sharp mode-locking regions appear only in small windows of intermittency. However, inspection \cref{fig:coevolutionary-asym-bifurcation-diagram-cascade}(b) suggests that chaotic regions between intermittencies feature a blurred mode-locking ratios, i.e., there are no sharp mode-locking steps due to chaos. As we decrease $\eta_\sigma$ further, the chaotic motion subsides together with the spiking modes, i.e., the spiking ratio drops to $R=0$. We leave a more detailed analysis for the chaotic region and its associated bifurcations for a future study.

\begin{figure}[htp!]
  \centering
  \begin{overpic}[width=\columnwidth]{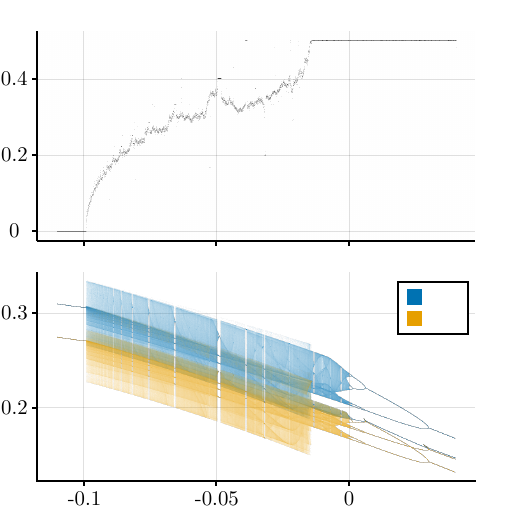}
  \put(-1,94){a)}
  \put(-1,45){b)}
  \put(60,94){$\downarrow$}
  \put(60,34){$\downarrow$}
    % Axis labels
    \put(6,48){\large\( \kappa \)}
    \put(95,5){\Large\( \eta_\sigma \)}
    \put(6,95){\large\( R \)}
    \put(95,52){\Large\( \eta_\sigma \)}
    % Legend labels
    \put(83.5,41){\(\max_t(\kappa)\)}
    \put(83.5,36.7){\(\min_t(\kappa)\)}
  \end{overpic}
  \caption{
    Spiking ratios $R$ (a) and the bifurcation diagram (b)
    for varying $\eta_\sigma$ with adaptivity \( a = 1 \) (gray dashed line in \cref{fig:coevolutionary-asym-ratio-many}) and fixed parameters \( \eta_\delta = 0.1 \), \( \varepsilon = 0.01 \) and \( b = 0 \). 
    The bifurcation diagram in panel (b) reports local minima and maxima in time, \( \min_t(\kappa) \)  and \( \max_t(\kappa) \).
    Both diagrams are produced via quasi-continuation from right to left.
  }
  \label{fig:coevolutionary-asym-bifurcation-diagram-cascade}
\end{figure}
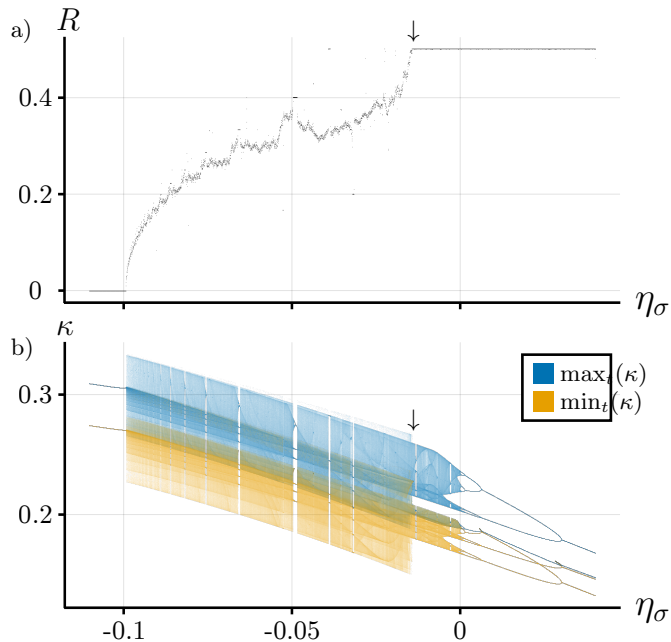

% ========================================================================

\section{Discussion}
We have investigated the dynamics of a pair of Theta neurons with non-autaptic coupling, i.e., there is only coupling between the two neurons without self-interaction. The first part of the study concerned the case of stationary coupling strengths, and the second part includes an ad-hoc adaptation rule for the coupling strength~\cite{Berner2019,juettner2023adaptive}, which effectively leads to the emergence of co-evolutionary network dynamics~\cite{gross2008adaptive,berner2023adaptive}.

The non-adaptive case, \( a=0 \), leads already to a surprisingly wide range of dynamical phenomena, when comparing \eqref{eq:reduced-model} to the Kuramoto model with $N=2$ oscillators~\cite{juettner2023adaptive}. The principal reason for this is that the Kuramoto model inherently exhibits a phase shift invariance; thus, a two oscillator model can be reduced to a one dimensional system 
\begin{align}\label{eq:KM_reduced_N2}
 \dot{\phi}&= \omega + K\sin{\phi},
\end{align}
describing the evolution of the phase difference $\phi(t)$ and frequency difference $\omega$ of the two oscillator nodes{, respectively}. Thus, the system exhibits only two simple behaviors, drift ($K<|\omega|$) and mode-locking ($K>|\omega|$), where only a 1:1 frequency ratio between oscillators is possible. Due to the excitable nature of the Theta neuron, they do not possess a phase shift symmetry; consequently, the system cannot be reduced and is genuinely two dimensional.

For the case of identical neurons, $\eta_\delta = 0$,  attractive states of quiescence (\quiescent\quiescent), librations (\quiescentlibration\quiescentlibration) and spiking in both neurons (\spiking\spiking) are observed. {These states transition through various bifurcation scenarios of the saddle-node type and display regions of bi-stability}.
{The identical nature of the neurons implies symmetric dynamics, i.e., the two neurons display identical behavior.
It is useful to see the parameter plane defined by \( \eta_\delta = 0 \) as the limiting collection of degenerate bifurcations for the more general case of non-identical neurons. Indeed, librations (around centers) and the spiking rotations are foliated in phase space; however, these dynamics are not structurally robust.}%, and break against small perturbations such as non-identical excitabilities, $\eta_\delta \neq 0$.
 
For non-identical neurons, $\eta_\delta \neq 0$, centers  transform into spirals and the librations are thus not any longer present; furthermore, the broken symmetry implies that neurons adhere to distinct dynamics --- e.g., we observe states where only the one neuron with larger excitability) is spiking (\quiescent\spiking).
In contrast to the identical case, bi-stable configurations are entirely absent.
% Mode-locking
In the regime where both neurons are {excited and} spiking (\spiking\spiking), the neuron with higher excitability appears to always spike at a higher rate. While the Kuramoto model~\eqref{eq:KM_reduced_N2} only allows for a single locking mode, non-identical excitabilities break the symmetry of the system resulting in a fan of locking regions with varying spiking ratios \( R \in \rationals\). As characteristic time scales of the individual phases are varied, either via $\kappa$ or $\eta_\sigma$, non-identical excitabilities allow spiking neurons to resonate in  modes corresponding to various frequency ratios. The resulting Arnol'd tongues form a devil's stair case, and, as we continuously increase the coupling strength, $\kappa$, the associated mode locking regions decrease their width, while the frequency ratio asymptotically converges towards $R \to 1$ --- i.e., a locking state, which most closely resembles the frequency locked state of the Kuramoto model.

Several new dynamical phenomena/features appear as we allow the coupling to adapt {with $a>0$}. 
The system is now three dimensional, thus inadvertently implying the possibility for richer dynamics.
%  bistability
First, regions of bi-stability are observed between quiescent (\quiescent\quiescent) and single-neuron spiking (\quiescent\spiking).
Second, the mode-locking regions gain in width as the adaptivity parameter $a$ is increased. For sufficiently large adaptivity, Arnol'd tongues begin to overlap, even multiple of them, so that topologically several modes of limit-cycles can exist at the same time --- in other words, the steps in the stair case begin to overlap. Thus, a very high multiplicity of mode-locking configurations is offered by the system to reside in. 
% modelocking 
As $a$ is further increased, oscillations in each locking region (with fixed $R$) undergo a period doubling cascade  to chaos with quasi-periodic orbits. In some parameter regions, chaotic attractor{s even} appear to collide in more complicated bifurcation scenarios. As a result, the stair case appears to get 'blurred' by chaotic dynamics (see Fig.~\ref{fig:coevolutionary-asym-bifurcation-diagram-cascade}).

It is interesting to {compare the non-adaptive and the adaptive dynamics as observed} in different oscillator models. As is the case in the present study, adaptive Kuramoto oscillators with $N=2$ display a widening of the locking region with increasing adaptivity $a$, together with bi-stability of drift and locked oscillator dynamics. While the Kuramoto model only offers a single control parameter with the frequency difference $\omega$, in the Theta neuron model we may chose from two parameters, the average and difference of excitabilities, $\eta_\sigma$  and $\eta_\delta$, respectively. {To investigate the (possible) widening of} locking regions, we chose $\eta_\sigma$ as the control parameter, since {fixing the} value $\eta_\delta\neq 0$ already guarantees the heterogeneity necessary to observe locking.
In simple terms, the nature of {adaptive coupling}, leading to co-evolutionary network dynamics, leaves more 'space' for mode-locking: {the coupling may adapt in a way that accommodates better for} heterogeneity, i.e., non-identical frequencies in the Kuramoto model and excitabilities in the Theta neuron model, respectively. Consequently, the coupling adapts while maintaining stability of the locked mode and thus 'stretches' the width of the locking region. This effect is seen in both models.
Finally, it is interesting to note --- in a certain sense --- a structural similarity between the adaptation rule in Eq.~\eqref{eq:adaptationrule1} and the phase dynamics of the Kuramoto model in Eq.~\eqref{eq:KM_reduced_N2}: in either of the two models, locked states imply a fixed (but adaptive) relationship between phase difference and coupling strength.

% Outlook: 
An open challenge presents itself with the question as of what dynamics result from the adaptive Theta neuron model~\eqref{eq:reduced-model} for large $N$. A recent study  performed a careful bifurcation analysis for the adaptive Kuramoto model~\cite{CestnikMartensAdaptive2024} under the two assumptions that $\epsilon$ be small and that $N\to\infty$. Using an (approximative) dimensional reduction of the system by averaging the rows of the associated coupling matrix, {the subsequent analysis was based on studying a} self-consistency equation. A similar approach for the system considered here seems momentarily unreachable, since the formulation of a self-consistency equation for Theta-neurons first would need to be shown to be feasible. However, the approximate dimensional reduction via row-averaging the coupling matrix is certainly feasible and may yield interesting insights into the dynamics. Generally speaking, a formidable challenge remains to develop other means for the dimensional reduction of adaptive oscillator networks and co-evolutionary network dynamics.
% TO BE DONE: What (approximate) reductions work for large N? (Concept of constrained adaptive models? C.f. ongoing study by Martens + Bick).

{A difficulty in studying systems with $N>2$ oscillators is to classify the mode-locking states in a meaningful way --- indeed, large oscillator systems are known to exhibit a multitude of frequency clusters~\cite{Berner2019}}. Following the approach of this study, i.e., using the adaptivity parameter $a$ to perturb a way from non-adaptive case, it makes sense to consider methods from the analysis of stationary networks with non-uniform coupling corresponding to the non-adaptive limit where \(a=0\). Indeed, the adaptive nature of the coupling may lead to simple stationary structures corresponding to submatrices (blocks) of uniform coupling within the coupling matrix, similar to models of coupled oscillators arranged in $M$ populations~\cite{bick2020understanding}. Indeed, such an approach was recently used~\cite{duchet_bick_2023} to study the effect of time-dependent or adaptive coupling in models with subpopulations, though with (non-adaptive) fixed size. Other authors have used methods from representation theory to determine frequency clusters in oscillator systems with stationary coupling~\cite{sorrentino2016complete}. 
It will be interesting to see if such methods may be adopted to adaptive network models. 
 
Finally, to make further progress in the mathematical theory of co-evolutionary network dynamics there is a certain need to clarify the relationship between the ad-hoc adaptation rule used in this study (which can be regarded as a truncated Fourier series) and more realistic plasticity models used in computational neuroscience such as spike-time-dependent plasticity~\cite{GerstnerKistler2002,YamakouJinjieMartens2024}. 
While the Theta neuron model lends itself for a simplistic formulation of the adaptation rule on the phase domain, while standard plasticity models for the equivalent quadratic integrate and fire model are formulated on the time domain, i.e., spikes are integrated on the time rather than the phase domain.
While some work has been formulated to relate the two dynamic representations~\cite{lucken2016noise}, a rigorous mathematical treatise appears to be missing. Alternatively, it could also be worthwhile to compare the dynamics resulting from the model studied here with the adaptive corollary of a QIF model in the large $N$ limit. Finally, the adaptation rule in~\eqref{eq:adaptationrule1} formulates a \emph{symmetric} response in terms of the phase difference; synaptic regulation more naturally adheres to an asymmetric response, respecting the different roles of pre- and post-synaptic dynamics. A future study could address the effect of such a model. 

% ----------
% b1) Phase difference model: spike train learning model? Can someone formulate a better adaptation model in time domain, combined with QIF? (Biophysically proper model ... Someone could / should do)
% b2) Does using a proper STDP model on time domain (potentiation and ...) yield similar bifurcation structures? (cf study with Yamakou). Maybe the similarities become apparent when comparing large N versions of the models.

\begin{acknowledgments}
We wish to acknowledge R. Cestnik and C. Bick for helpful comments and discussions.
We gratefully acknowledge financial support from the Royal Swedish Physiographic Society of Lund.
\end{acknowledgments}

\appendix

% ========================================================================
\section{Solutions of a quadratic in an interval} \label{app:polynomial-in-interval}
Here we present a derivation of the number of (fixed point) solutions to \cref{eq:fixedpointquadratic}.
We begin with a general lemma on solutions of a quadratic in an interval:
\begin{lemma} \label{lem:polynomial-in-interval}
  The polynomial
  \begin{equation} \label{eq:lemma-polynomial}
    p(x) = a x^2 + b x + 1
  \end{equation}
  has the following number and locations of solutions in the interval \( x \in [0, 1] \):

  \begin{table}[htp!]
    \begin{center}
      \begin{tabular}{llll}
        \toprule
        \multicolumn{3}{c}{Criteria} & Solutions \\
        \midrule
        \( a + b + 1 < 0 \)  & & & 1 solution \( x \in (0, 1) \) \\
        \( a + b + 1 = 0 \), & \( a \leq 1 \) & & 1 solution \( x = 1 \) \\
                             & \( a > 1 \) & & 2 sols. \( x_1 \in (0, 1) \), \( x_2 = 1 \) \\
        \( a + b + 1 > 0 \), & \( b \geq -2 \) & & No solution \\
                             & \( b < -2 \), & \( 4a > b^2 \) & No solution \\
                             &          & \( 4a = b^2 \) & 1 solution \( x \in (0, 1) \) \\
                             &          & \( 4a < b^2 \) & 2 solutions \( x_1, x_2 \in (0, 1) \) \\
        \bottomrule
      \end{tabular}
    \end{center}
  \end{table}

  Note that all \( a,b \in \reals \) are covered in the table above.
\end{lemma}

\begin{proof}
  The proof goes through a set of assumptions on \( a \) and \( b \) that cover all \( a, b \in \reals \).
  Where the assumptions does not map one-to-one with the criteria in the table, an explanation is given.

  \paragraph{Assume that \( p(1) = a + b + 1 < 0 \).}
  Since \( p(0) = 1 > 0 \), by the intermediate value theorem, there must be at least one solution in \( (0, 1) \).
  Because \( p \) is quadratic, there can not be two roots in the interval and while the function has different signs on the boundary.
  \( x = 1 \) is not a root.
  Thus there is only one root \( x \in (0, 1) \).

  \paragraph{Assume that \( p(1) = a + b + 1 = 0 \) and \( a \leq 0 \).}
  \( x = 1 \) is a root.
  Also,
  \begin{equation}
    p(x) > 0,\, \forall x \in (0, 1).
  \end{equation}
  \( x = 1 \) is thus the only root.

  \paragraph{Assume that \( p(1) = a + b + 1 = 0 \) and \( 0 < a \leq 1 \).}
  \( x = 1 \) is a root and \( p \) has a minimum in
  \begin{equation}
    x_\text{min} = -\frac{b}{2a} = \frac{1 + a}{2a} \geq 1,
  \end{equation}
  so
  \begin{equation}
    p(x) > 0,\, \forall x \in (0, 1).
  \end{equation}
  Thus \( x = 1 \) is the only root.

  \paragraph{Assume that \( p(1) = a + b + 1 = 0 \) and \( 1 < a \).}
  \( x_2 = 1 \) is a root and \( p \) has a minimum in
  \begin{equation}
    x_\text{min} = -\frac{b}{2a} = \frac{1 + a}{2a}.
  \end{equation}
  Since
  \begin{equation}
    0 < x_\text{min} < 1
  \end{equation}
  and
  \begin{equation}
    p(x_\text{min}) = - \frac{(a - 1)^2}{4a} < 0,
  \end{equation}
  by the intermediate value theorem there is another root \( x_1 \in (0, x_\text{min}) \subset (0, 1) \).

  \paragraph{Assume that \( p(1) = a + b + 1 > 0 \) and \( a \leq 0 \).}
  Then one root \( x_1 < 0 \) and the other root \( x_2 > 1 \) (if \( a = 0 \), only one of these roots exist).
  Thus there are no roots in \( [0, 1] \).
  Note that the assumption lies in only one region in the table, since \( b > -a - 1 \geq -1 > -2 \).

  \paragraph{Assume that \( a + b + 1 > 0 \), \( a > 0 \) and \( b \geq 0 \).}
  The minimum
  \begin{equation}
    x_\text{min} = -\frac{b}{2a} \leq 0,
  \end{equation}
  and thus
  \begin{equation}
    p(x) > 0,\, \forall x \in [0, 1].
  \end{equation}

  \paragraph{Assume that \( p(1) = a + b + 1 > 0 \), \( a > 0 \) and \( -2a < b < 0 \).}
  The minimum
  \begin{equation}
    x_\text{min} = -\frac{b}{2a} \in (0, 1).
  \end{equation}
  Thus there will be no, one or two solutions in \( (0, 1) \) when
  \begin{align}
    p(x_\text{min}) &= 1 - \frac{b^2}{4a} > 0 \iff 4a > b^2 \\
    p(x_\text{min}) &= 1 - \frac{b^2}{4a} = 0 \iff 4a = b^2 \\
    p(x_\text{min}) &= 1 - \frac{b^2}{4a} < 0 \iff 4a < b^2
  \end{align}
  respectively.

  \paragraph{Assume that \( p(1) = a + b + 1 > 0 \), \( a > 0 \) and \( b \leq -2a \).}
  The minimum
  \begin{equation}
    x_\text{min} = -\frac{b}{2a} \geq 1,
  \end{equation}
  and thus
  \begin{equation}
    p(x) > 0,\, \forall x \in [0, 1].
  \end{equation}
  Note that \( 2b > -2a - 2 \geq b - 2 \) or equivalently \( b \geq -2 \).
\end{proof}

The coefficients in \cref{eq:fixedpointquadratic} are
\begin{subequations}
  \begin{align}
    a &= -2 \kappa \frac{\nu_1}{\nu_2} \\
    b &= \nu_1
  \end{align}
\end{subequations}
when \( \abs{\nu_1} \leq \abs{\nu_2} \) and
\begin{subequations}
  \begin{align}
    a &= -2 \kappa \frac{\nu_2}{\nu_1} \\
    b &= \nu_2
  \end{align}
\end{subequations}
when \( \abs{\nu_2} \leq \abs{\nu_1} \).
We apply the lemma separately for the identical and non-identical case.
In the case of identical neurons, we also have to consider the case \( \nu_1 = \nu_2 = \eta + 2\kappa - 1 = 0 \) in order to cover the entire parameter space.
The number of fixed points in different regions can be seen in \cref{tab:nonadaptive-sym-fixed-points,tab:nonadaptive-asym-fixed-points}.

\begin{table}[htp!]
  \begin{center}
    \begin{tabular}{lll}
      \toprule
      \multicolumn{2}{c}{Criteria} & Num. fixed points \\
      \midrule
      \( \eta < 0 \), & \( \eta + 2 \kappa - 1 \neq 0 \) & 4 fixed points \\
                      & \( \eta + 2 \kappa - 1 = 0 \) & Continuum where \( 2 \kappa x y = 1 \) \\
      \( \eta = 0 \), & \( \kappa < - \frac{1}{2} \) & 5 fixed points \\
                      & \( \kappa \geq - \frac{1}{2} \) & 1 fixed point \\
      \( \eta > 0 \), & \( \kappa < - \frac{\left( 1 + \sqrt{\eta} \right)^2}{2} \) & 8 fixed points \\
                      & \( \kappa = - \frac{\left( 1 + \sqrt{\eta} \right)^2}{2} \) & 4 fixed points \\
                      & \( \kappa > - \frac{\left( 1 + \sqrt{\eta} \right)^2}{2} \) & No fixed points \\
      \bottomrule
    \end{tabular}
  \end{center}
  \caption{Fixed points for identical excitabilities, \(\eta=\eta_1=\eta_2\).}
  \label{tab:nonadaptive-sym-fixed-points}
\end{table}

\begin{table*}[htp!]
  \begin{center}
    \begin{tabular}{llll}
      \toprule
      \multicolumn{3}{c}{Criteria} & Num. fixed points \\
      \midrule
      \( \nu_1 > 0 \), & \( \nu_1 \nu_2 - 2 \kappa \nu_1 + \nu_2 < 0 \) & & 4 fixed points \\
              & \( \nu_1 \nu_2 - 2 \kappa \nu_1 + \nu_2 = 0 \) & & 2 fixed points \\
      \( \nu_2 < 0 \), & \( \nu_1 \nu_2 - 2 \kappa \nu_2 + \nu_1 > 0 \) & & 4 fixed points \\
              & \( \nu_1 \nu_2 - 2 \kappa \nu_2 + \nu_1 = 0 \), & \( \nu_1 + 2 \kappa \nu_2 \leq 0 \) & 2 fixed points \\
              & & \( \nu_1 + 2 \kappa \nu_2 > 0 \) & 6 fixed points \\
      \( \nu_2 < -2 \), & \( \nu_1 \nu_2 - 2 \kappa \nu_2 + \nu_1 < 0 \), & \( \nu_1 \nu_2 + 8 \kappa < 0 \) & 8 fixed points \\
              & & \( \nu_1 \nu_2 + 8 \kappa = 0 \) & 4 fixed points \\
      otherwise      & & & No fixed points \\
      \bottomrule
    \end{tabular}
  \end{center}
  \caption{Fixed points for non-identical excitabilities, \( \eta_1 \neq \eta_2 \).}
  \label{tab:nonadaptive-asym-fixed-points}
\end{table*}

% ========================================================================
% \clearpage
\section{Details of numerics}
Below follows some notes on the numerics used in the paper.

\subsection{Numerical scheme for finding attractors} \label{app:attractor-scheme}
For a parameter value, initial conditions are sampled throughout the phase space.
For each such initial condition, the model is integrated using an explicit 5/4 Runge--Kutta method for some time \( \transienttime \) where no states except the final one is stored.
Then, the integration is restarted with this state as the initial value and integrated using an explicit 5/4 Runge--Kutta method for some time \( \stabletime \), and the state in every time step is recorded.
Based on the assumption that all transient behavior occurred during the first \( \transienttime \) part of the integration, the initial condition is categorized as being in the basin of attraction for a type of attractor based on the following criteria:
\begin{enumerate}
  \item If the state has changed more than \( 2\pi \) in any \( \theta_k \), the orbit is labeled a rotation in those \( \theta_k \).
  \item If the state has changed less than some tolerance \( \fixedpointtolerance \) in the 2-norm, the orbit is labeled a fixed point.
  \item Otherwise, the orbit is labeled a libration.
\end{enumerate}
The parameter value is then marked as having at least one attractor of that given type.
After the process has been repeated for all initial conditions, the found types of attractors for the parameter value can then be recorded.

The method can miss types of attractors if the sampling is not dense enough; at least one sample per basin of attraction is required to get correct results.
The method can also falsely identify some types of orbits if the time \( \transienttime \) is shorter than the duration of transient behavior.
False identification can also happen if the time \( \stabletime \) is shorter than the period of a rotation.

Also note that this scheme ignores how many distinct attractors of a given type exist, as well as differentiating between periodic and quasi-periodic rotations.

% \subsection{A numerical scheme for approximating spike ratios} \label{app:ratio-scheme}
% \hllong{%
%   Go through algorithm step by step.
% }
% \hllong{%
%   Describe what/why I am meassuring.
% 
%   Relate procedure to that for attractors.
% 
%   Mention limitation of basins of attraction.
% }

\subsection{Modification of co-evolutionary system for MatCont} \label{app:coev-in-matcont}
Bifurcation curves of the co-evolutionary system were investigated using the numerical continuation software MatCont \cite{dhooge2008matcont}.
MatCont requires that the dynamical system has a real vector space \( \reals^m \) as the phase space.
The embedding \( \theta_k \mapsto \left( \cos \theta_k, \sin \theta_k \right) = \left( x_k, y_k \right) \) of the periodic phase variables is used to achieve this.
The vector field is extended in the radial direction by
\begin{equation}
  \dt{r_k} = 1 - r_k
\end{equation}
where \( x_k = r_k \cos \theta_k \) and \( y_k = r_k \sin \theta_k \) such that the unit circle is attractive.
Thus the node dynamics used in MatCont are
\begin{subequations}
  \begin{align}
    \begin{split}
      \dt{x_k} &= \diff*{\left( r_k \cos \theta_k \right)}{t} \\
      &= \dt{r_k} \cos \theta_k - r_k \sin(\theta_k) \dt{\theta_k} \\
      &= (1 - r_k) \cos \theta_k - y_k \dt{\theta_k}
    \end{split}\\
    \begin{split}
      \dt{y_k} &= \diff*{\left( r_k \sin \theta_k \right)}{t} \\
      &= \dt{r_k} \sin \theta_k + r_k \cos(\theta_k) \dt{\theta_k} \\
      &= (1 - r_k) \sin \theta_k + x_k \dt{\theta_k}
    \end{split}
  \end{align}
\end{subequations}
where
\begin{equation}
  r_k = \sqrt{x_k^2 + y_k^2}
\end{equation}
and
\begin{equation}
  \cos \theta_k = \frac{x_k}{r_k}.
\end{equation}

\section*{References}
\nocite{*}
% \bibliography{refs.bib}% 
%aipnum4-2.bst 2019-01-14 (MD) hand-edited version of apsrev4-1.bst

\end{document}